\documentclass[preprint,12pt]{elsarticle}
\usepackage[utf8]{inputenc}
\usepackage{multirow}
\usepackage{multicol}
\usepackage[ruled,vlined,linesnumbered]{algorithm2e}

\usepackage{hyperref}
\usepackage{array}
\usepackage{graphicx}
\usepackage{math_definition}
\usepackage{float}
\usepackage{amssymb}
\usepackage{svg}

\journal{the Journal of Symbolic Computation}

\begin{document}

\begin{frontmatter}
\title{Isolating Bounded and Unbounded Real Roots of a Mixed Trigonometric-Polynomial}
\author[label1]{Rizeng Chen}
\ead{xiaxueqaq@stu.pku.edu.cn}
\author[label1]{Haokun Li}
\ead{haokunli@pku.edu.cn}
\author[label1]{Bican Xia}
\ead{xbc@math.pku.edu.cn}
\author[label1]{Tianqi Zhao}
\ead{zhaotq@pku.edu.cn}
\author[label2]{Tao Zheng}
\ead{taozheng@amss.ac.cn}
\affiliation[label1]{organization={School of Mathematical Sciences, Peking University},
            city={Beijing},
            country={China}}
\affiliation[label2]{organization={Key Laboratory of Mathematics Mechanization, Mathematics and Systems Science, Chinese Academy of Sciences},
            city={Beijing},
            country={China}}
\begin{abstract}
Mixed trigonometric-polynomials (MTPs) are functions of the form $f(x,\sin{x},\allowbreak \cos{x})$ with $f\in\Q[x_1,x_2,x_3]$. In this paper, an algorithm ``isolating" all the real roots of an MTP is provided and implemented. It automatically divides the real roots into two parts: one consists of finitely many ``bounded" roots in an interval $[\mu_-,\mu_+]$ while the other consists of probably countably many ``periodic" roots in $\mathbb{R}\backslash[\mu_-,\mu_+]$. For bounded roots, the algorithm returns isolating intervals and corresponding  multiplicities while for periodic roots, it returns finitely many mutually disjoint small intervals $I_i\subset[-\pi,\pi]$, integers $c_i>0$ and multisets of root multiplicity $\{m_{j,i}\}_{j=1}^{c_i}$ such that any periodic root $t>\mu_+$ is in the set $(\sqcup_i\cup_{k\in\mathbb{N}}(I_i+2k\pi))$ and any interval $I_i+2k\pi\subset(\mu_+,\infty)$ contains exactly $c_i$ periodic roots with multiplicities $m_{1,i},...,m_{c_i,i}$, respectively. The effectiveness and efficiency of the algorithm are shown by experiments.
Besides, the method used to isolate the roots in $[\mu_-,\mu_+]$ is applicable to any other bounded interval as well. The algorithm takes advantages of the weak Fourier sequence technique and deals with the intervals period-by-period without scaling the coordinate so to keep the length of the sequence short. The new approaches can easily be modified to decide whether there is any root, or whether there are infinitely many roots in unbounded intervals of the form $(-\infty,a)$ or $(a,\infty)$ with $a\in\Q$.
\end{abstract}

\end{frontmatter}
\section{Introduction}
Real root isolation has been considered as a symbolic way to compute the real roots of a univariate polynomial or function. Real root isolation for polynomials is of fundamental importance in computational real algebraic geometry (\emph{e.g.}, it is a routing sub-algorithm for cylindrical algebraic decomposition \cite{collins1975quantifier}), and there are many efficient algorithms and popular tools in computer algebra systems such as Maple and Mathematica to isolate the real roots of polynomials.

There are also some research on real root isolation for non-polynomial functions: In \cite{mitchell1990robust}, interval arithmetic has been used to isolate real roots of a large class of analytic functions. The algorithm there mainly considers deciding the exclusion or the inclusion of a root in an interval. Recent years, Strzebo{\'n}ski presented a real root isolation procedure for univariate \emph{exp–log–arctan functions} \cite{strzebonski2008real,strzebonski2012real}.
The procedure is based on the technique of weak Fourier sequence and semi-Fourier sequence. Meanwhile, Strzebo{\'n}ski studied real root isolation for the so-called tame elementary functions using the concept of false derivative \cite{strzebonski2009real}. The completeness of these algorithms of Strzebo{\'n}ski relies on Schanuel's conjecture. 
In \cite{achatz2008deciding,mccallum2012deciding}, for developing a decision procedure for a certain class of
sentences of first order logic involving integral polynomials and
a certain specific analytic transcendental function trans$(x)$, Achatz, McCallum and Weispfenning provided a recursive real root isolation algorithm for  a certain kind of
generalised integral polynomial in trans$(x)$: $f(x,\text{trans}(x))$ (where $f(x, y)$ is a given polynomial in $y$ whose coefficients are elements of the ring of fractions of $\mathbb{Z}[x]$ with respect to powers of a specific
integral polynomial $d(x)$), using \emph{pseudo-differentiation} and
Rolle’s theorem in the spirit of \cite{collins1976polynomial}, and also a classical result of
Lindemann \cite{shidlovskii2011transcendental}. The list of
transcendental functions to which their decision method directly
applies includes $\exp(x),\ln(x)$ and $\arctan(x)$. The decision procedure when trans$(x)=\exp(x)$ is implemented and the corresponding root isolation algorithm is believed to be more efficient than the one 
provided in \cite{maignan1998solving}. Inspired by \cite{achatz2008deciding} and \cite{strzebonski2011cylindrical}, Xu \emph{et al.} studied a class of quantified exponential polynomial formulas extending polynomial ones \cite{xu2015quantifier}. 
The isolation algorithm for an exponential polynomial used therein is modified from the one in \cite{achatz2008deciding}.
In \cite{huang2018positive}, Huang \emph{et al.} developed some algorithms for isolating the positive roots of a class of functions called poly-powers, while in \cite{wang2021symbolic}, Wang and Xu focused on isolating the real roots of a certain class of radical expressions. 


A {\em mixed trigonometric-polynomial} (MTP) refers to a function of the form $f(x,\sin x, \cos x)$ with $f$ being a trivariate polynomial having rational coefficients. In this paper, we present an algorithm ``isolating" all the real roots (there are probably infinitely many of them) of such an MTP. While all inverse trigonometrical functions are exp-log-arctan functions, this is not true for the MTPs. In fact, since an exp-log-arctan function has finitely many real roots (\cite[Claim 4]{strzebonski2012real}), any MTP with infinitely many real roots is not an exp-log-arctan function (\emph{e.g.}, $\sin x$, $\cos x$, $2x\sin x+1$, \emph{etc.}). On the other hand, an MTP is usually not a tame elementary function since the variable varies over the whole $\R$. Hence, our result is not covered by the ones of Strzebo{\'n}ski. However, in Section \ref{sec:IBR}, while isolating the real roots of an MTP in a bounded interval, we take advantages of the weak Fourier sequences of a function introduced in \cite{strzebonski2008real,strzebonski2012real} (it can be also practicable to use the false derivative technique in \cite{strzebonski2009real} since any MTP restricted to a bounded interval is tame). Nevertheless, we do not need Schanuel's conjecture since we only consider trigonometric functions in one variable. Note that both the methods used in \cite{strzebonski2012real} and \cite{mccallum2012deciding} are applicable for isolating the real roots of the function of the form $f(\arctan x,x)$, with $f$ being a bivariate polynomial having rational coefficients. This is equivalent to isolating the real roots of the function $f(\hat{x},\tan \hat{x})$ with $\hat{x}\in(-\frac{\pi}{2},\frac{\pi}{2})$ (see, \emph{e.g.}, \cite[Section 6]{mccallum2012deciding}). The goal of the present paper is to isolate all the real roots of an MTP, 
which is a non-trivial generalisation.

The topic of automatically proving inequalities involving trigonometric functions, which is closely related to the topic of the present paper, has also drawn more and more attention these years due to their applications to engineering problems and different branches of science \cite{chen2020automated}. Usually, the term ``MTP inequalities" refers to inequalities of the form $g(x)>0$ or $g(x)\geq0$ with $g$ being an MTP. A so-called ``Natural Approach" introduced in \cite{mortici2011natural} by Mortici, approximates $\sin x$ and $\cos x$ by their Taylor polynomials and reduce the problem of proving an MTP inequality to the problem of proving some polynomial (or rational function) inequalities. That approach leads to a series of subsequent studies \cite{chen2012sharp,bercu2016pade,bercu2017natural,makragic2017method,chenliu2016}, in which a great deal of inequalities involving trigonometric functions are solved automatically. Later in  \cite{chen2020automated}, Chen and Liu provided a complete algorithm proving MTP inequalities and discussed systematically, for the first time, the termination of the algorithm. However, these algorithms only prove MTP inequalities over bounded intervals (\emph{e.g.}, $(0,\frac{\pi}{2})$ or $(0,\frac{\pi}{4}]$, \emph{etc.}). Till recently, Chen and Ge \cite{chen2022automated} tried to develop a procedure automatically proving MTP inequalities on the unbounded interval $(0,\infty)$. Unfortunately, we found some mistakes therein. Therefore, the problems of proving an MTP inequality, and of isolating real roots in an unbounded interval remain open so far. 

In this paper, an algorithm ``isolating" all the real roots of an MTP is provided (Algorithm \ref{alg:complete-algorithm}).  It automatically divides the real roots into two parts: one consists of finitely many ``bounded" roots in an interval $[\mu_-,\mu_+]$ with $\mu_-<0<\mu_+$, while the other consists of probably countably many ``periodic" roots in $\mathbb{R}\backslash[\mu_-,\mu_+]$. For the bounded roots, the algorithm returns some isolating intervals and the multiplicity of each of them. For those periodic roots greater than $\mu_+$, it returns finitely many mutually disjoint small intervals $I_i\subset[-\pi,\pi]$, integers $c_i>0$ and multisets of root multiplicity $\{m_{j,i}\}_{j=1}^{c_i}$ such that any periodic root $t>\mu_+$ is in the set $(\sqcup_i\cup_{k\in\mathbb{N}}(I_i+2k\pi))$ and any interval $I_i+2k\pi\subset(\mu_+,\infty)$ contains exactly $c_i$ periodic roots with multiplicities $m_{1,i},...,m_{c_i,i},$ respectively. For those periodic roots less than $\mu_-$, it returns some similar results. The algorithm is implemented with {\tt Maple2021}. A large number of examples (either from the literature or randomly generated) have been tested to show its effectiveness and efficiency.

In particular, this is the first algorithm to decide whether or not a given MTP has infinitely many real roots: it does iff the algorithm returns at least one periodic root. And of course, the algorithm is also able to decide whether or not an MTP has a real root (when there are real roots, we can compute the multiplicity of all of them). This allows us to prove or disprove an MTP inequality over $\mathbb{R}$. Our results also indicate, for the first time, that the ``distributions" of the roots of an MTP in the ``periods" $(-\pi,\pi]+2k\pi$, with $k>0$ (\emph{resp.}, $k<0$) sufficiently large, share a same pattern, which seems to be a basic law about the real roots of MTPs.

Besides, the method used to isolate the roots in $[\mu_-,\mu_+]$ is applicable to any other bounded interval as well. It takes advantages of the weak Fourier sequence technique and deals with the intervals period-by-period without scaling the coordinate (see, \emph{e.g.}, \cite[Algorithm 4.1, Step 2]{chen2022automated}) so to keep the length of the sequence short. An interesting problem is to compare the efficiency of the false derivative technique with the weak Fourier technique for isolating real roots of an MTP in a bounded interval, which is beyond the scope of this paper.

As a direct application, we discuss briefly at the end of the paper how one can modified our new approaches to prove or disprove an MTP inequality on an interval of the form $(a,b)$, $[a,b)$, $(a,b]$, $[a,b]$, with $a,b\in\Q\cup\{\pm\infty\}$.

The rest of the paper is organized as follows.
The next section introduces the problem and some basic results that will be used from time to time. In Section \ref{sec:farroots}, we introduce the algorithm ``isolating" periodic roots while in Section \ref{sec:IBR} we provide the procedure computing the isolating intervals for bounded roots. The main algorithm ``isolating" all real roots of an MTP is presented in Section \ref{sec:wholealgorithm}. And in Section \ref{sec:experiments}, experimental results are provided to indicate the effectiveness and efficiency of the algorithm. Finally, Section \ref{sec:conclusion} summarizes the paper and provides some remarks relating to the topic of proving MTP inequalities.

\section{The Problem}\label{sec:Preprocessing}
\subsection{Some basic results}
Denote by $\mathbb{N}, \Z, \Q, \R$ and $\C$ the sets of natural, integer, rational, real and complex numbers, respectively. For any subset $A\subset\mathbb{R}$, set $A_{>0}=\{a\in A\,|\,a>0\}$. The meanings of the notation $A_{\geq0}$, $A_{\leq0}$ and $A_{<0}$ are similar. For any $\mathbb{A}\in\{\Z, \Q, \R,\C\}$, $\mathbb{A}[x_1,\ldots,x_m]$ denotes the corresponding ring of polynomials in the variables $x_1,\ldots,x_m$ with coefficients in $\mathbb{A}$. For any $f\in\mathbb{A}[x_1,\ldots,x_m]$, $\deg(f,x_i)$, $\lc(f,x_i)$ and $\discrim(f,x_i)$ denote the degree, the leading coefficient and the discriminant of $f$ with respect to $x_i$, respectively. Moreover, $\tdeg(f,[x_{i_1},\ldots,x_{i_{t}}])$ denotes the total degree of the subsequence of variables $[x_{i_1},\ldots,x_{i_{t}}]$.

{\bf The problem} considered in this paper is to solve the real roots of the following equation:
\begin{equation}\label{eq:goalequation}
f(x,\sin{x},\cos{x})=0,~{\rm where}~f\in \Q[x,y,z].
\end{equation}
We provide some basic observations in the following.

\begin{lemma}\label{lemma:NoLimit}
For any nonzero $f\in\Q[x,y,z]$, the real root set of $f(x,\sin{x},\cos{x})$  has no accumulation point in $\mathbb{R}$.
\end{lemma}
\begin{proof}
The complex-valued function $h(z)=f(z,\sin{z},\cos{z})$ is analytic over $\C$. By definition, $h|_{\mathbb{R}}=f(x,\sin{x},\cos{x})$. If the set of the real roots of $f$ has an accumulation point in $\mathbb{R}$, then the set of the complex roots of $h$ has an accumulation point in $\mathbb{C}$ as well, which indicates $h\equiv0$ over $\mathbb{C}$. But $h$ can not be identically $0$ over $\mathbb{C}$, since $f$ is not over $\mathbb{R}$.
\end{proof}

\begin{lemma}\label{lemma:piRealRoot}
For any $q(x)\in\Q[x]$, ``$r\pi$ is a real root of $q(x)$ for any $r\in\Q\setminus\{0\}$" iff ``$q(x)$ is the zero polynomial in $\Q[x]$".
\end{lemma}
\begin{proof}
The ``if" part is obvious. The ``only if" part follows from the fact that $r\pi$ is transcendental for any $r\in\Q\setminus\{0\}$.
\end{proof}

\begin{theorem}\label{thm:BothAlgeNum}\cite[Chapter 2.3]{shafarevich1998number}
    If $\alpha\in\C$ and $\alpha\ne0$, then there is at least one transcendental number in each pair $(\alpha,\sin\alpha)$, $(\alpha,\cos\alpha)$, $(\alpha,\tan\alpha)$,
    $(\alpha,\cot\alpha)$.\qed
\end{theorem}
\begin{definition}\label{definition:recip}
For a nonzero $g\in\Q[x,y]$, we define its \emph{reciprocal polynomial} $\mathcal{R}(g)$ as $g(x,\frac{1}{y})y^{\deg(g,y)}$. Additionally, we set $\mathcal{R}(0)=0$.
\end{definition}
From the definition it is obvious that 

\;\;1) $\mathcal{R}(g_1g_2)=\mathcal{R}(g_1)\mathcal{R}(g_2)$ for any $g_1,g_2\in\Q[x,y]$, 

\;\;2) if $g=\sum_{i=0}^{\deg(g,y)}a_i(x)y^i$ with $a_0(x)\neq0$ (\emph{i.e.}, $y\nmid g$), then $\mathcal{R}(\mathcal{R}(g))=g$, and

\;\;3) $y$ does not divide $\mathcal{R}(g)$ for any nonzero $g$.


\noindent{}We then have some key properties of $\mathcal{R}$ as follows.

\begin{proposition}\label{proposition:RkeepIrr}
Suppose $g\in\Q[x,y]$ and $g\neq cy$ for any $c\in\mathbb{Q}$. Then ``$g$ is irreducible in $\Q[x,y]$" implies ``$\mathcal{R}(g)$ is irreducible in $\Q[x,y]$".
\end{proposition}
\begin{proof}
The proof is inductive on the total degree of $g$: If $\tdeg(g,[x,y])=0$, the conclusion is trivially true since $g$ is not irreducible. Suppose that $\tdeg(g,[x,y])=n>0$ and the implication holds for any polynomial ($\neq cy$) of total degree less that $n$. Assume $\mathcal{R}(g)$ is reducible and its factorization is $c_0h_1^{p_1}\cdots h_t^{p_t}$ with $h_i$ irreducible, $\sum_{j=1}^{t}p_j\geq2$ and $c_0\in\mathbb{Q}\backslash\{0\}$. Moreover, each $h_i\neq cy$ for any nonzero $c\in\Q$ because $y\nmid\mathcal{R}(g)$. Since $g$ is irreducible and $g\neq cy$ for any $c\in\Q$, $y\nmid g$. Thus $g=\mathcal{R}(\mathcal{R}(g))=c_0\mathcal{R}(h_1)^{p_1}\cdots \mathcal{R}(h_t)^{p_t}$. By the inductive hypothesis, each $\mathcal{R}(h_i)$ is irreducible. Since $\sum_{j=1}^{t}p_j\geq2$, $g$ is reducible, which is a contradiction.
\end{proof}

For any polynomial $g\in\Q[x,y]$ of positive total degree, $g$ is said to be \emph{square-free}, if the factorization of $g$ is $cg_{1}\cdots g_{t}$ with $c\in\Q\setminus\{0\}$,  $t\in\Z_{\ge1}$, each $g_i$ irreducible and $g_i$ and $g_j$ co-prime for any $1\le i\ne j\le t$. Any constant polynomial is not square-free in our definition.

\begin{proposition}\label{proposition:recipKeepSquarefree}
Suppose $g\in\Q[x,y]$ and $g\neq cy$ for any $c\in\mathbb{Q}$. Then ``$g$ is square-free" implies ``$\mathcal{R}(g)$ is square-free".
\end{proposition}
\begin{proof}
Factorize $g$ as $c_0y^\ell h_1\cdots h_t$, with $c_0\in\Q$ nonzero, $\ell=0\text{ or }1$ (corresponding to the cases $y\nmid g$ and $y\,|\,g$ respectively), $t\geq1$ (since $g\neq cy$ and $g\notin\Q$) and each $h_i$ irreducible. 

First, $\mathcal{R}(g)=c_0\mathcal{R}(h_1)\cdots\mathcal{R}(h_t)$ since $\mathcal{R}(y^\ell)=1$. 
Second, for each $h_i$, $h_i\neq cy$ for any $c\in\Q$ because $y\nmid h_i$.
So, by Proposition \ref{proposition:RkeepIrr}, $\mathcal{R}(h_i)$ is irreducible.
Finally, $\mathcal{R}(h_i)$ and $\mathcal{R}(h_j)$ are co-prime for $1\leq i\neq j\leq t$. Suppose they are not, then $\mathcal{R}(h_i)=c'\mathcal{R}(h_j)$ for a $c'\in\Q\backslash\{0\}$. Since $y\nmid h_i$ and $y\nmid h_j$, $h_i=\mathcal{R}(\mathcal{R}(h_i))=c'\mathcal{R}(\mathcal{R}(h_j))=c'h_j$, contradicting the fact that $h_i$ and $h_j$ are co-prime in the factorization of $g$.
\end{proof}

\subsection{Decomposing the problem into sub-problems}
Because
\[
\left\{
\begin{array}{rl}\vspace{2mm}
    \sin x&= \frac{2\tan{\frac{x}{2}}}{1+\tan^2(\frac{x}{2})}\\
    \cos x&= \frac{1-\tan^2{\frac{x}{2}}}{1+\tan^2({\frac{x}{2}})}
\end{array}\right.,\;\;x\neq 2k\pi+\pi~(k\in \Z),
\]
Eq. \eqref{eq:goalequation} is equivalent to
\begin{subnumcases}
{}
 f(x,\sin{x},\cos{x})=0, ~~~~~~~~~~~~~~~~~~~~~x=2k\pi+\pi\label{meq:1},\\
    f(x,\frac{2\tan{\frac{x}{2}}}{1+\tan^2{(\frac{x}{2})}},\frac{1-\tan^2{(\frac{x}{2})}}{1+\tan^2{(\frac{x}{2}})})=0,~  \hfill x\in (2k\pi-\pi,2k\pi+\pi), \label{meq:2_}
\end{subnumcases}
where Eq. \eqref{meq:1} is equivalent to $f(x,0,-1)=0,~x=(2k+1)\pi$.
By Lemma \ref{lemma:piRealRoot}, 
$(2k+1)\pi$ is a root of $f(x,0,-1)$ iff $f(x,0,-1)$ is the zero polynomial in $\Q[x]$.
Through cancelling the denominators, we define a polynomial as follows:
\begin{equation}\label{equation:gTurnedToTan}
    g(x,t)=f(2x,\frac{2t}{1+t^2},\frac{1-t^2}{1+t^2})\cdot(1+t^2)^{\tdeg(f,[y,z])}\in\mathbb{Q}[x,t].
\end{equation}
Since $1+\tan^2{(\frac{x}{2})}>0$, solving the real roots of Eq. (\ref{meq:2_}) is then equivalent to solving the real roots of the equation
\begin{equation}
    g(x,\tan{x})=0,~x\in(k\pi-\frac{\pi}{2},k\pi+\frac{\pi}{2}).\label{meq:2}
\end{equation}

\begin{example}\label{example:gtan}
Consider the following MTP
$$f(x,\sin{x},\cos{x})=x\cos^4{x}-2\sin{x}\cos^3{x}+\frac{4}{5}\sin^3{x}.$$
By Eq. \eqref{equation:gTurnedToTan}, we have 
\begin{align*}
g(x,t)=10xt^8+20t^7-40xt^6-28t^5+60xt^4+92t^3-40xt^2-20t+10x.
\end{align*}
\end{example}

In Eq. \eqref{meq:2}, we have to deal with infinitely many intervals $(k\pi-\frac{\pi}{2},k\pi+\frac{\pi}{2})$ while $k$ varies in $\Z$. Fortunately, Proposition \ref{prop:firstbound} and Corollary \ref{cor:main_case} below, together with
Theorem \ref{thm:location-of-the-roots} in the next section will set us free from doing that.
\begin{proposition}\label{prop:firstbound}
Set nonzero polynomial $g\in\Q[x,y]$, $\delta>0$ and $U\subset\mathbb{R}$ to be a bounded set. If $|\lc(g,x)|>\delta$ for any $y\in U$, then there is an integer $N>0$ such that $x>N$ implies $\forall y\in U,\; g(x,y)\neq0$.
\end{proposition}
\begin{proof}
Suppose $\deg(g,x)=n$ and $a_i(y)\in\mathbb{Q}[y]\;(i=0,\ldots,n)$ are the coefficients of $g$ \emph{w.r.t.} $x$. Now that $U$ is bounded, there is $B_i>0$ such that $|a_i(y)|<B_i$ for any $y\in U$. For sufficiently large $x$, we have $\delta |x|^n>\sum_{i=0}^{n-1}B_i|x|^i$. Therefore, for any $y\in U$,
\[|a_n(y)x^n|>\delta|x|^n>\sum\limits_{i=0}^{n-1}B_i|x|^i\geq|\sum\limits_{i=0}^{n-1}a_i(y)x^i|.\]
This indicates that $g(x,y)\neq0$ for any $y\in U$.
\end{proof}


By Proposition \ref{prop:firstbound}, the value of the tangent function at every ``periodic'' real root of $g(x,\tan{x})$ appears near a real root of $\lc(g,x)$.  
In order to locate these ``periodic'' real roots, we introduce the following definition.

\begin{definition}\label{def:potential-periodic-interval}
For any $g\in\Q[x,y]$ with $\deg(g,y)\ge 1$,
let $I$ be a finite list of open intervals $[(a_0,b_0),\ldots,(a_{s+1},b_{s+1})]$, 
where $s\in\Z_{\ge0}$, $a_0=-\infty$, $b_{s+1}=+\infty$, $b_0<0$, $a_{s+1}>0$, $a_j\in \Q~(1\le j\le s+1)$ and $b_j\in \Q~(0\le j\le s)$, such that
\begin{enumerate}[label=$($\roman*$)$]
\item  for each $j~(1\le j\le s)$, $a_j< b_j$ and $b_j \le a_{j+1}$, 
\item for each real root of $\lc(g,x)$, there exists an interval $(a_j,b_j)~(1\le j\le s)$ containing it $($there is no root of $\lc(g,x)$ in $(a_0,b_0)$ and $(a_{s+1},b_{s+1}))$,
\item \label{item:potential-periodic-interval-3} each interval $(a_j,b_j)~(1\le j\le s)$ has exactly one real root of $\lc(g,x)$.
\end{enumerate}
We call $I$ a \emph{potential periodic interval set} of $g$.

Generally, if we replace ``exactly" in \ref{item:potential-periodic-interval-3} by ``at most", we call $I$ a \emph{general potential periodic interval set} of $g$.
\end{definition}

\begin{example}\label{example:potential-periodic-interval-set}
Consider the polynomial in Example \ref{example:gtan}:
\begin{align*}
g(x,y)=10xy^8+20y^7-40xy^6-28y^5+60xy^4+92y^3-40xy^2-20y+10x.   
\end{align*}
We have $\lc(g,x)=10y^8-40y^6+60y^4-40y^2+10$. 
It has two real roots $-1$ and $1$ and the isolating intervals of them are $(-\frac{215}{128}, -\frac{19}{32})$ and $(\frac{19}{32}, \frac{215}{128})$, respectively.
Thus,
{\footnotesize
$$\left[\left(-\infty, -\frac{63}{16}\right), \left(-\frac{215}{128}, -\frac{19}{32}\right), \left(\frac{19}{32}, \frac{215}{128}\right), \left(\frac{63}{16}, +\infty\right)\right]$$}is a potential periodic interval set of $g$.
\end{example}

\begin{corollary}\label{cor:main_case}
Let $g\in\Q[x,y]$ be with $\deg(g,y)\ge 1$, $[(a_0,b_0),\ldots,(a_{s+1},b_{s+1})]$ be a general potential periodic interval set of $g$, and
\begin{align}
    U:=\R\setminus\bigcup_{j=0}^{s+1}(a_i,b_i).
\end{align}
Then, there exists a positive integer $N$ such that 
$$
\forall x\in \R ~\forall y\in U~(x>N \Rightarrow g(x,y)\neq 0).
$$
\end{corollary}
\begin{proof}
Note that $|\lc(g,x)|$ has a lower bound on $y\in U$ and $U$ is bounded.
Then, the conclusion follows from Proposition \ref{prop:firstbound}.
\end{proof}


For the polynomial $g\in\Q[x,y]$ in Eq. \eqref{meq:2}, 
let $U,a_j,b_j$ be the same as in Corollary \ref{cor:main_case}.
Take $N_+, N_-\in\Z_{\ge1}$ such that $\forall x\in \R ~\forall y\in U~(x>N_+ \Rightarrow g(x,y)\neq 0)$ and
$\forall x\in \R ~\forall y\in U~(x<-N_- \Rightarrow g(x,y)\neq 0)$.
There exist $k_+\in\Z_{\ge0}$ and $k_-\in\Z_{\le0}$ such that
$k_+\pi-\frac{\pi}{2}<N_+<k_+\pi+\frac{\pi}{2}$ and
$k_-\pi-\frac{\pi}{2}<-N_-<k_-\pi+\frac{\pi}{2}$, respectively.
Thus, Eq. \eqref{meq:2} can be transformed equivalently into the following three cases:

\begin{enumerate}[label=$($\ref{meq:2}\alph*$)$]
    \item \label{item:MainThreeCase1}
        \begin{align*}
            g(x,\tan x)=0,~~~~&x\in(k\pi-\frac{\pi}{2},k\pi+\frac{\pi}{2})
            ~~~~~~~~~~~~~~~~~~~~~~~~~~~~~~~~~~~~~~~~
        \end{align*}
        where $k_{-}\leq k \leq k_{+}\;(k\in\Z)$;
    \item \label{item:MainThreeCase2}
    \begin{align*}
        g(x,\tan{x})=0,~~~~&x\in(k\pi-\frac{\pi}{2},k\pi+\arctan(b_0))\\
                         &x\in(k\pi+\arctan(a_i),k\pi+\arctan(b_i)),~{i=1,\ldots,s}\\
                         &x\in(k\pi+\arctan(a_{s+1}),k\pi+\frac{\pi}{2})
    \end{align*}
        where $k>k_{+}\;(k\in\Z)$;
    \item \label{item:MainThreeCase3}
    \begin{align*}
        g(x,\tan{x})=0,~~~~&x\in(k\pi-\frac{\pi}{2},k\pi+\arctan(b_0))\\
                         &x\in(k\pi+\arctan(a_i),k\pi+\arctan(b_i)),~{i=1,\ldots,s}\\
                         &x\in(k\pi+\arctan(a_{s+1}),k\pi+\frac{\pi}{2})
    \end{align*}
        where $k<k_{-}\;(k\in\Z)$.
\end{enumerate}

A possible way to solve case \ref{item:MainThreeCase1} for $k_-\le k\le k_+$ is to go back to the original MTP $f(x,\sin x,\cos x)$ and consider the roots in $(2k_-\pi-\pi,2k_+\pi+\pi)$. Replace $x$ in $f(x,\sin x,\cos x)$ by $\ell x$ with $\ell\in\Z_{>0}$ large enough such that the roots in $(2k_-\pi-\pi,2k_+\pi+\pi)$ ``shrink" to the roots of $f(\ell x,\sin \ell x,\cos \ell x)$ in $(-\pi/2,\pi/2)$. Then unify the sine and cosine functions with the tangent function in $f(\ell x,\sin \ell x,\cos \ell x)$ to obtain some $\hat{g}(x,\tan x)$. Finally, deal with the roots of $\hat{g}(-x, -\tan x)$ and $\hat{g}(x, \tan x)$ in the interval $(0,\pi/4)$. 
The reason why we do not use this method in Section \ref{sec:IBR} is that $\ell$  may be large and the ``degree" of $\tan x$ in $\hat{g}(x, \tan x)$ can be high. Instead, we will deal with case \ref{item:MainThreeCase1} period-by-period to avoid that problem.

The last case can be reduced to the second last one via replacing $g(x,\tan x)$ by $g(-x,-\tan x)$. In the next section, we focus on the asymptotic properties of the real roots corresponding to case \ref{item:MainThreeCase2}.

\section{Asymptotic Property of Roots}\label{sec:farroots}
In this section, we shall deal with case \ref{item:MainThreeCase2}, namely isolate the ``periodic'' roots of $g(x,\tan{x})$ where $g$ is a nonzero polynomial in $\Q[x,y]$.

\subsection{Detecting and locating periodic roots}

There are countably many ``periodic'' roots of $g(x,\tan{x})$, so they can not be located by the methods of traditional real root isolation.
However, these ``periodic'' roots can be divided into finitely many bunches via a general potential periodic interval set (recall Def. \ref{def:potential-periodic-interval}).
The roots in each bunch appear periodically, that is why we call them ``periodic'' roots.
In this subsection, we propose an algorithm to isolate these ``periodic'' roots by means of a general potential periodic interval set.

First, we introduce the definition of \emph{analytic delineable} and an important result about it from the work of McCallum \cite{mccallum1998improved}.

\begin{definition}[\cite{mccallum1998improved}]
Let $\vX$ denote the $(r-1)$-tuple $(x_1,\cdots,x_{r-1})$. An $r$-variate polynomial $f(\vX,x_r)$ over the reals is said to be \emph{analytic delineable} on a sub-manifold $S$ $($usually connected$)$ of $\R^{r-1}$, if 
\begin{enumerate}[label=$($\arabic*$)$]
    \item the portion of the real variety of $f$ that lies in the cylinder $S\times \R$ over $S$ consists of the union of the graphs of some $k\geq 0$ analytic functions $\theta_1<\cdots <\theta_k$ from $S$ into $\R$, and
    \item there exist positive integers $m_1$, $\cdots$, $m_k$ such that for every $a\in S$, the multiplicity of the root $\theta_i(a)$ of $f(a,x_r)$ $($considered as a univariate polynomial in $x_r$$)$ is $m_i$.
\end{enumerate}
Also, an $r$-variate polynomial $f(\vX,x_r)$ over $\R$ is said to be \emph{degree-invariant} on a subset $S$ of $\R^{r-1}$ if the total degree of $f(p,x_r)$ $($as a univariate polynomial in $x_r$$)$ is the same for every point $p$ of S.
\end{definition}

\begin{proposition}[{\cite[Theorem 2]{mccallum1998improved}}]\label{prop:McCallum}
Let $r\geq 2$ be an integer and $f(\vX,x_r)$ be a polynomial in $\R[\vX,x_r]$ of positive total degree. 
Suppose that $D(\vX):=\discrim(f,x_r)$ is a nonzero polynomial. 
Let $S$ be a connected submanifold of $\R^{r-1}$ such that $f$ is degree-invariant and does not vanish identically on $S$, and $D(\vX)>0$ or $D(\vX)<0$ over $S$. Then, $f$ is analytic delineable on $S$.\qed
\end{proposition}

We restate a well-known result (which follows from Proposition \ref{prop:McCallum}) in Corollary \ref{corollary:bound-M} and redefine the concept of ``root function'' in Definition \ref{def:root-function-start-point}.

\begin{corollary}\label{corollary:bound-M}
Let $g$ be a square-free polynomial in $\Q[x,y]$ with $\deg(g,y)\ge 1$, 
and $M$ be the maximal real root of the polynomial 
$\lc(g,y)\cdot\discrim(g,y)$.
Then,
\begin{enumerate}
\item $g$ is analytic delineable on $(M, +\infty)$, and
\item there exist some analytic functions $\theta_1(x)<\cdots<\theta_r(x)\allowbreak\;(r\ge0)$ defined on $(M, +\infty)$ such that for any $x_0\in(M, +\infty)$, $g(x_0,\theta_i(x_0))=0$ and $\theta_1(x_0),\ldots,\theta_r(x_0)$ are all real roots of $g(x_0,y)$. \qed
\end{enumerate}

\end{corollary}
\begin{definition}\label{def:root-function-start-point}
Keep the notation in Corollary \ref{corollary:bound-M}.
We call those analytic functions $\theta_1(x),\ldots,\theta_r(x)$ the
\emph{terminal-root-functions} of $g$ and $M$ the \emph{start-point} of the terminal-root-functions of $g$.
\end{definition}

\begin{example}\label{example:terminal-root-function}
Consider the polynomial $g(x,y)$ in Example \ref{example:potential-periodic-interval-set}.
We have $\lc(g,y)=-10x$ and
$\discrim(g,y)=-703687441776640000\cdot(163840000x^8 + 261529600x^6 + 83912256x^4 - 13333725x^2 - 10857025)$.
After isolating all real roots of them, we see that the maximal real root $M$ of $\discrim(g,y)$, which is in $(0,1)$, is the start-point of the terminal-root-functions of $g$.
And there are two terminal-root-functions of $g$, which are defined on $(M,+\infty)$ $($see Fig. \ref{fig:graph-Terminal-Root-Function}$)$.
\begin{figure}[ht]
    \centering
    \includegraphics[width=0.6\textwidth]{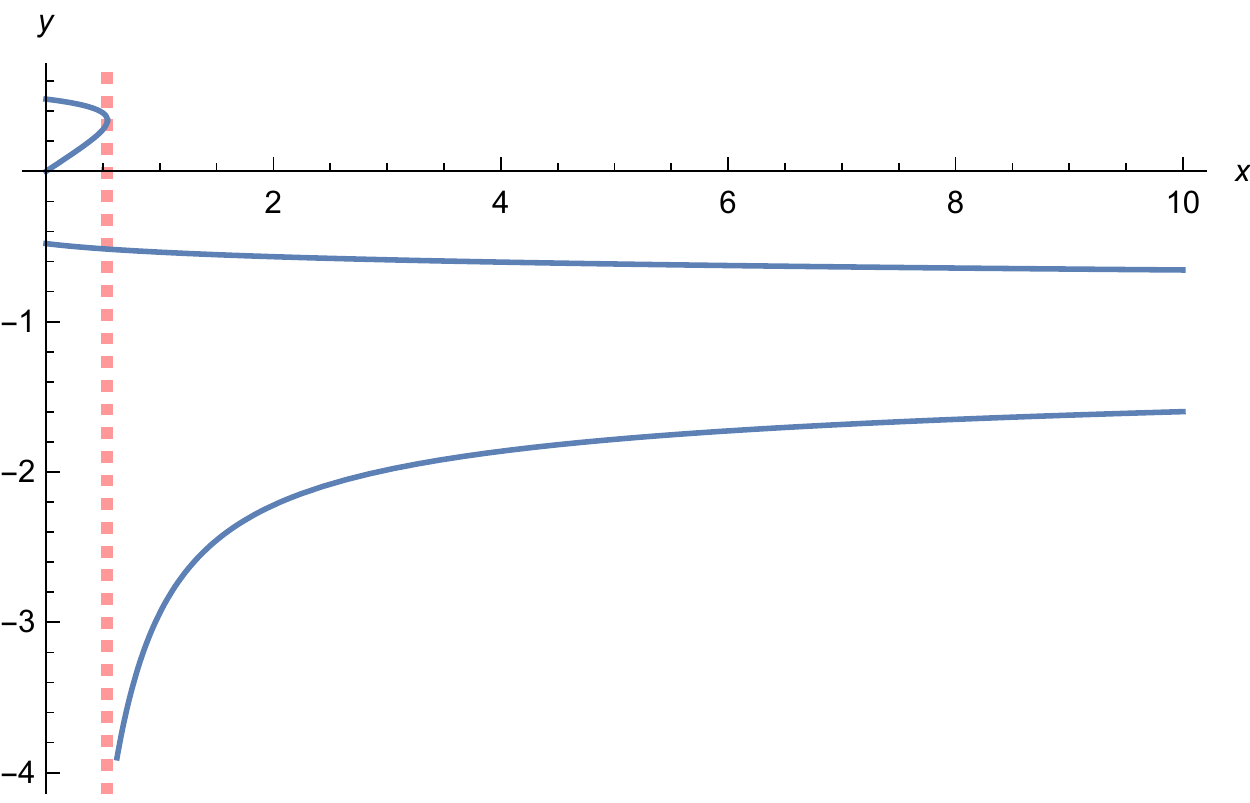}
    \caption{The figure of the terminal-root-functions of $g$. 
    The red dotted line shows the graph of the start-point $x=M$ and the graphs of the terminal-root-functions are plotted as blue curves.
    }
    \label{fig:graph-Terminal-Root-Function}
\end{figure}

\end{example}

In Theorem \ref{thm:location-of-the-roots} below, we compute a bound that guarantees well-defined terminal-root-functions $($as in Corollary \ref{corollary:bound-M}$)$, which are essential to locating the unbounded roots.

\begin{theorem}\label{thm:location-of-the-roots}
Let $g$ be a square-free polynomial in $\Q[x,y]$ with $\deg(g,y)\ge 1$, 
$[(a_0,b_0),\ldots,(a_{s+1},b_{s+1})]$ be a general potential periodic interval set of $g$, 
$k'$ be a natural number such that $M'=(k'+1/2)\pi$ is an upper bound of the start-point $M$ of the terminal-root-functions of $g$ and all real roots of the polynomial
$$\prod_{j=0}^{s}g(x,b_j)g(x,a_{j+1}),$$
and $\theta_1(x),\ldots,\theta_r(x)$ be the terminal-root-functions of $g$.
Then:
\begin{enumerate}[label=$($\arabic*$)$]
\item\label{item:thm:location-of-the-roots1} For  each $i~(1\le i \le r)$, there exists
$j_i~(0\le j_i \le s+1)$ such that $\theta_i(x)\in (a_{j_i},b_{j_i})$ whenever $x\in (M',+\infty)$.
\item\label{item:thm:location-of-the-roots2} 
For each real root $\xi$ of $g(x,\tan x)$ greater than $M'$, there exist $k>k'~(k\in \N)$ and $i$ $(1\le i \le r)$ such that
$$\xi\in\left(\arctan(a_{j_i})+k\pi, \arctan(b_{j_i})+k\pi\right).$$
\item\label{item:thm:location-of-the-roots3} 
For each $i~(1\le i \le r)$ and any $k>k'~(k\in \N)$, there exists  $$\xi\in\left(\arctan(a_{j_i})+\allowbreak k\pi, \arctan(b_{j_i})+k\pi\right)$$ such that   $\tan {\xi}=\theta_i(\xi)$.
\end{enumerate}
\end{theorem}

\begin{example}\label{example:thm:location-of-the-roots}
Consider the polynomial $g(x,y)$ in Example \ref{example:potential-periodic-interval-set}.
In Example \ref{example:potential-periodic-interval-set}, we compute a potential periodic interval set of $g(x,y)$:
{\footnotesize
$$\left[\left(-\infty, -\frac{63}{16}\right), \left(-\frac{215}{128}, -\frac{19}{32}\right),\allowbreak \left(\frac{19}{32}, \frac{215}{128}\right), \left(\frac{63}{16}, +\infty\right)\right].$$}And the isolating intervals of the real roots of the polynomial
{\footnotesize
$$g(x,-\frac{63}{16})\cdot g(x,-\frac{215}{128})\cdot g(x,-\frac{19}{32})\cdot g(x,\frac{19}{32})\cdot g(x,\frac{215}{128})\cdot g(x,\frac{63}{16})$$}are $(-8, -7)$, $(-4, -3)$, $(-1, 0)$, $(0, 1)$, $(3, 4)$ and $(7, 8)$.
In Example \ref{example:terminal-root-function}, we see that the start-point $M$ of the terminal-root-functions of $g$ is less than $1$.
Thus, we take $k'=3$ such that $(k'+1/2)\pi$ is an upper bound of $M$ and the real roots computed above. 
Note that the two terminal-root-functions of $g$ are bounded by $(-\frac{215}{128},-\frac{19}{32})$ when $x\in((k'+1/2)\pi,+\infty)$ $($see Fig. \ref{fig:thm:location-of-the-roots}$)$.
\begin{figure}[ht]
    \centering
    \includegraphics[width=0.6\textwidth]{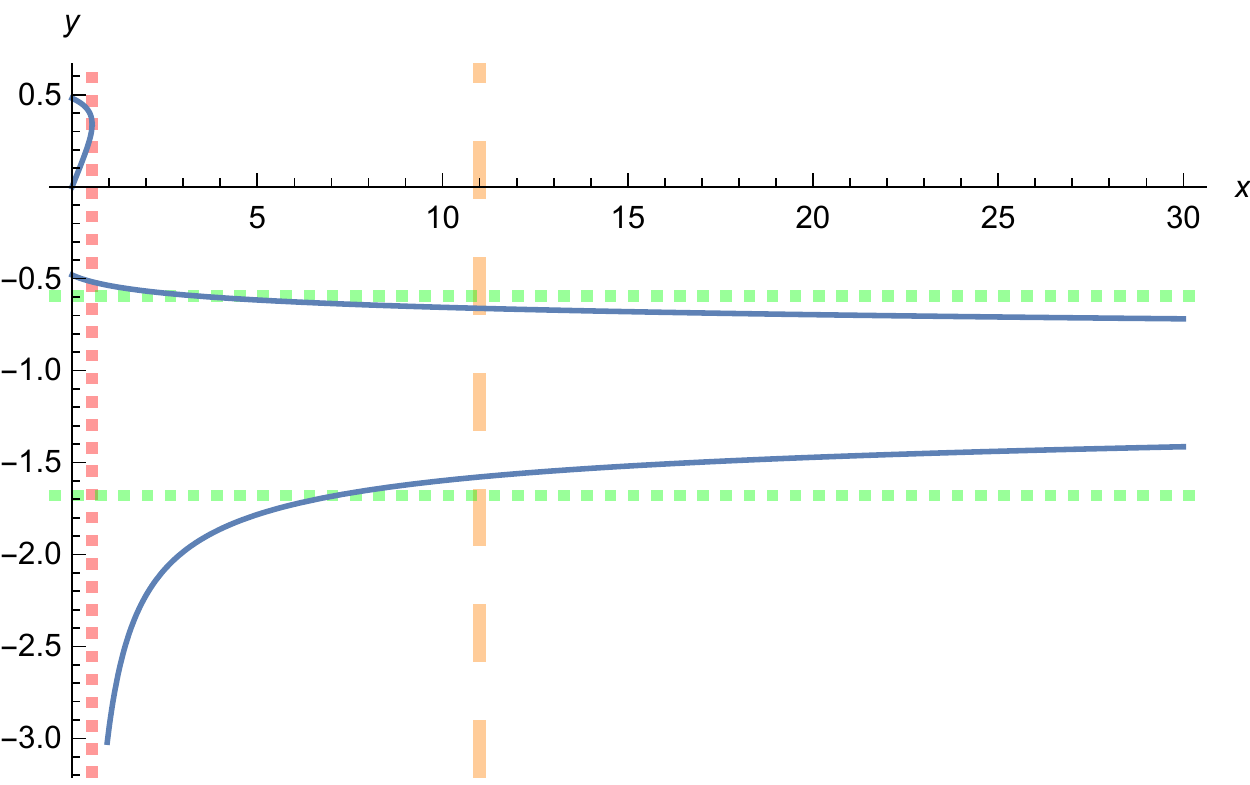}
    \caption{The figure of the terminal-root-functions of $g$ bounded by $(-\frac{215}{128},-\frac{19}{32}
    )$. 
    The red dotted line and the orange dashed line (which are both vertical) show the graph of the start-point $x=M$ and that of $x=(k'+1/2)\pi$, respectively.
    The two green dotted lines (which are horizontal) show the graphs of $y=-\frac{215}{128}$ and $y=-\frac{19}{32}$.
    When $x>(k'+1/2)\pi$, the terminal-root-functions (plotted as blue curves) are both bounded by the interval $(-\frac{215}{128},-\frac{19}{32}
    )$.
    }
    \label{fig:thm:location-of-the-roots}
\end{figure}

So, for every $k > 3~(k \in \Z)$, there are at least two real roots of $g(x,\tan{x})$ in $(k\pi-\arctan{\frac{215}{128}},k\pi-\arctan{\frac{19}{32}})$.
\end{example}
\begin{remark}\label{remark:decideexistence}
For each interval $(a_j,b_j)~(0\le j\le s+1)$, it may happen that none of the terminal-root-functions $\theta_i(x)$ is bounded by $(a_j,b_j)$ when $x>M'$. It is also possible that there can be one or more than one terminal-root-functions $\theta_{i_1}(x),\ldots,\theta_{i_t}(x)$, with $t\in\Z_{\ge1}$, that are bounded by $(a_j,b_j)$ when $x>M'$. By \ref{item:thm:location-of-the-roots2} and \ref{item:thm:location-of-the-roots3} in Theorem \ref{thm:location-of-the-roots}, there is no real root in the set $(M',\infty)\cap\tan^{-1}(a_j,b_j)$ in the former case, while in the latter case, there are at least $t$ real roots in each interval $(\arctan a_j+k\pi,\arctan b_j+k\pi)$ contained in $(M',\infty)$ 
$($in Line \ref{Line:alg1-detect-implicit-function} of Algorithm \ref{alg:AlgFarZero}, we distinguish these two cases$)$.
Denote these roots by $x^{(1)}_k<\cdots<x^{(t'_k)}_k$ for some $t'_k\geq t$. Since Theorem \ref{thm:location-of-the-roots} is valid for isolating intervals $(a_j,b_j)$ of arbitrarily small ``arc-tangent-length" $($\emph{i.e.}, $\arctan b_j-\arctan a_j)$ and any $M'$ large enough, we have
\[\lim_{k\rightarrow\infty}(x^{(t'_k)}_k-x^{(1)}_k)=0.\] This indicates that it is impossible to find real numbers $c_1,c_2\in(-\pi/2,\pi/2)$ such that they isolate, for instance, $x^{(1)}_k, x^{(2)}_k$ and $x^{(3)}_k$ $($if $t>2)$ uniformly. That is, there are no such $c_1,c_2$ satisfying  \[x^{(1)}_k<c_1+k\pi<x^{(2)}_k<c_2+k\pi<x^{(3)}_k\] for all sufficiently large integer $k$. This is why we only ``isolate" them with the interval $(\arctan a_j,\arctan b_j)+k\pi$ in the sense that every $x^{(\iota)}_k$ is contained in it. Nevertheless, in the next subsection we will prove that $t'_k=t$ for sufficiently large $k$, which means we can compute the exact number $t$ of real roots in the interval $(\arctan a_j,\arctan b_j)+k\pi$ that is sufficiently far from zero.
\end{remark}

\noindent{}\emph{\bf Proof of Theorem \ref{thm:location-of-the-roots}.}
\ref{item:thm:location-of-the-roots1} 
Note that $g(x,a_j)~(1\le j \le s+1)$ and $g(x,b_j)~(0\le j \le s)$ have no roots in $(M',+\infty)$.
Then, for every $i~(1\le i\le r)$ and for any $x\in(M',+\infty)$, $\theta_i(x)\not\in\{b_0,a_1,b_1,\ldots,b_s,a_{s+1}\}$.
By Proposition \ref{prop:McCallum}, $\theta_i(x)$ is continuous.
So, there exists $j_i~(0\le j_i\le s+1)$ such that $\theta_i((M',+\infty))\subseteq (a_{j_i},b_{j_i})$ or 
there exists $j_i~(0\le j_i\le s)$ such that $\theta_i((M',+\infty))\subseteq(b_{j_i},a_{j_i+1})$.
By Corollary \ref{cor:main_case}, for sufficiently large $x$, $\theta_i(x)\in \cup_{j=0}^{s+1}(a_j,b_j)$. Therefore, $\theta_i((M',+\infty))\subseteq (a_{j_i},b_{j_i})$ for some $0\leq j_i\leq s+1$.

\ref{item:thm:location-of-the-roots2}
If $\xi$ is a real root of $g(x,\tan{x})$ such that $\xi>M'$, then there exists $i~(1\le i\le r)$ such that 
$\theta_{i}(\xi)=\tan{\xi}$. By \ref{item:thm:location-of-the-roots1}, $\theta_i((M',+\infty))\subseteq (a_{j_i},b_{j_i})$. So, $\tan{\xi}\in (a_{j_i},b_{j_i})$. Thus, there exists $k~(k>k')$ such that $\xi\in (\arctan(a_{j_i})+k\pi,\arctan(b_{j_i})+k\pi)$. 

\ref{item:thm:location-of-the-roots3}
By \ref{item:thm:location-of-the-roots1}, for each $i~(1\le i\le r)$, $\theta_i((M',+\infty))\subseteq (a_{j_i},b_{j_i})$. So, for any $k~(k>k')$, we have
\begin{align*}
    \theta_i(\arctan(a_{j_i})+k\pi)>a_{j_i}=&\tan(\arctan(a_{j_i})+k\pi),\\
    \theta_i(\arctan(b_{j_i})+k\pi)<b_{j_i}=&\tan(\arctan(b_{j_i})+k\pi).
\end{align*}
Hence there exists $\xi\in(\arctan(a_{j_i})+k\pi, \arctan(b_{j_i})+k\pi) $ such that $\tan \xi-\theta_i(\xi)=0$. 
$\hfill{\square}$\par

\vspace{3mm}

\begin{algorithm}[!ht]
\scriptsize
\DontPrintSemicolon
\LinesNumbered
\SetKwInOut{Input}{Input}
\SetKwInOut{Output}{Output}
\Input{ 
$g$, a square-free polynomial in $\Q[x,y]$ such that $\deg(g,y)\geq 1$;\\
$I=[(a_0,b_0),\ldots,(a_{s+1},b_{s+1})]$, a  general potential periodic interval set of $g$.
}
\Output{
$k'$, a natural number;\\
$ret$, a finite set of pairs $((x_j^{-},x_j^{+}),c_j)$ such that for any $k>k'~(k\in\N)$, there are \emph{AT LEAST} $c_i$ roots for $g(x,\tan x)$ in the interval $(x_j^{-}+k\pi,x_j^{+}+k\pi)$.
}
\caption{\bf IsolatingFarZero}\label{alg:AlgFarZero}
\BlankLine
Find $k'\in\N$ such that
$(k'+\frac{1}{2})\pi$ is an upper bound for all real roots of $\lc(g,y)\cdot\discrim(g,y)\cdot\prod_{j=0}^{s}g(x,b_{j})g(x,a_{j+1})$. \label{Line:alg1-computation-of-M}\;
Choose an integer $x_0$ such that $x_0>(k'+\frac{1}{2})\pi$\label{Line:alg1-sample-point}.\; 
$c_j\leftarrow0$ for $j=0,\ldots,s+1$\label{Line:alg1-ini-cj}\;
\For{each real root $r$ of $g(x_0,y)$}{\label{Line:alg1-for-1}
    \For{$j$ from $0$ to $s+1$}   
    {
        \If{$r\in (a_j,b_j)$}{
            $c_{j}\leftarrow c_{j}+1$\label{Line:alg1-count-c}\;
            \textbf{break}\label{Line:alg1-for-1-end}\;
        }
    }
}
$ret\leftarrow \emptyset$\label{Line:alg1-organizing-output}\;
\For{$j$ from $0$ to $s+1$\label{Line:alg1-detect-implicit-function-1}}{
    \If{$c_j>0$\label{Line:alg1-detect-implicit-function}}{
    $x_{j}^{-}\leftarrow\arctan a_j,~x_{j}^{+}\leftarrow\arctan b_j$\;
		$ret \leftarrow ret \cup \{((x_{j}^{-},x_{j}^{+}),c_j)\}$\label{Line:alg1-output-end}\;
    }
}
\Return $k'$, $ret$\;
\end{algorithm}

\begin{definition}
Let $g$, $k'$ be the same as in Theorem \ref{thm:location-of-the-roots}, and $\theta(x)$ be a terminal-root-function of $g(x,y)$. 
We call the set
$\{x_0\in ((k'+1/2)\pi,+\infty) \mid \theta(x_0)=\tan{x_0} \}$
\emph{a bunch of periodic roots} of $g(x,\tan{x})$. 
\end{definition}

Summarizing all above, we propose Algorithm \ref{alg:AlgFarZero} to isolate the ``periodic" roots of $g(x,\tan{x})$.
Given a square-free polynomial $g\in\Q[x,y]$ with $\deg(g,y)\geq1$, and a general potential periodic interval set $I$ of $g$, Algorithm \ref{alg:AlgFarZero} computes a natural number $k'$ and a set of pairs $((x_{i}^{-},x_{i}^{+}),c_i)$, such that for any $k>k'~(k\in\N)$, there are at least $c_i$ real roots of $g(x,\tan x)$ in the interval $(x_{i}^{-}+k\pi,x_{i}^{+}+k\pi)$.

\begin{theorem}\label{thm:AlgFarZeroisCorrect}
Algorithm \ref{alg:AlgFarZero} is correct.
\end{theorem}
\begin{proof}
The process of Algorithm \ref{alg:AlgFarZero} is based on Theorem \ref{thm:location-of-the-roots}.
And it is clear that Algorithm \ref{alg:AlgFarZero} is correct by the proof of Theorem \ref{thm:location-of-the-roots}.
\end{proof}

\subsection{Ensuring exact quantity of periodic roots in an ``isolating" interval}\label{subsec:exactnumber}

In the subsection, we would strengthen Theorem \ref{thm:location-of-the-roots}: for all sufficiently large $k$, $\theta_i(x)$ and $\tan x$ intersects exactly once in the interval $(\arctan a_{j_i}+k\pi,\arctan b_{j_i}+k\pi)$. And a lower bound for such $k$ can be efficiently determined.


Depending on whether the derivative $\theta_i'(x)\ (x>M')$ is strictly bounded by $1$ or not,
the discussion is divided into two cases.
For the bounded case, the key observation is: the derivative of $\tan x$ is $1/\cos^2 x$, which is always $\geq 1$. Therefore, since $\theta_i'(x)$ is smaller than $1$, $y=\tan x$ can not intersect with $\theta_i(x)$ twice in a period $(k\pi -\pi/2,k\pi+\pi/2)$, where $k\in\N$ and $(k+1/2)\pi>M'$. 
For the unbounded case, we consider $1/\theta_i(x)$ instead of $\theta_i(x)$.

We first establish a lemma.
\begin{lemma}\label{lemma:existence-of-delta}
Let $g\in\Q[x,y]$ be square-free with $\deg(g,y)\geq1$.
There exists a finite set $S\subset\Q$ such that for any $\delta_0\in\Q\setminus S$, 
$\res(\frac{\partial g}{\partial x}+\delta_0 \frac{\partial g}{\partial y},g,y)$
is not the zero polynomial in $\Q[x]$.
\end{lemma}
\begin{proof}
Regard $\delta$ as a variable. We claim that 
$\res(\frac{\partial g}{\partial x}+\delta \frac{\partial g}{\partial y},g,y)$ is not the zero polynomial in $\Q[\delta,x]$.
Note that $\frac{\partial g}{\partial y}$ is not the zero polynomial, since $\deg(g,y)\ge1$.
Then, $q:=\gcd(\frac{\partial g}{\partial x}+\delta \frac{\partial g}{\partial y},~g)$ is a nonzero polynomial in $\Q[x,y]$.
We also note that $q\mid\left(\frac{\partial g}{\partial x}+\delta \frac{\partial g}{\partial y}\right)$ implies $q\mid\frac{\partial g}{\partial y}$. Now that $g$ is square-free, $q\mid g$ and $q\mid\frac{\partial g}{\partial y}$ imply $\deg(q,y)=0$.  Thus, $\res(\frac{\partial g}{\partial x}+\delta \frac{\partial g}{\partial y},g,y)$ is not the zero polynomial in $\Q[\delta,x]$.
Hence we can let $S$ be the set of the common roots in $\Q$ of the coefficients of  $\res(\frac{\partial g}{\partial x}+\delta \frac{\partial g}{\partial y},g,y)$ with respect to $x$. The proof is then completed.
\end{proof}

\begin{lemma}\label{thm:derivative-bound}
Let $g\in\Q[x,y]$ be a square-free polynomial with $\deg(g,y)\geq1$,
$\delta_0$ a rational number such that $\res(\frac{\partial g}{\partial x}+\delta_0 \frac{\partial g}{\partial y},g,y)$ is not the zero polynomial,
 $\theta_1(x),\ldots,\theta_r(x)$ the terminal-root-functions of $g$, and $\hat{k}\in\N$ a natural number such that  $\hat{M}=(\hat{k}+1/2)\pi$ is an upper bound of the start-point of the terminal-root-functions of $g$ and all real roots of   $\res(\frac{\partial g}{\partial x}+\delta_0 \frac{\partial g}{\partial y},g,y)$. 
Then, the derivative of each $\theta_i(x)$ satisfies either one of the following properties:
\begin{enumerate}[label=$($\arabic*$)$]
    \item\label{item:thm:derivative-bound1} $\forall x>\hat{M}:$ $\theta_i'(x)<\delta_0$;
    \item\label{item:thm:derivative-bound2} $\forall x>\hat{M}:$ $\theta_i'(x)>\delta_0$.
\end{enumerate}
We call $\hat{k}$ a \emph{$\delta_0$-bound-point} of $g$.
In particular, if $\delta_0<1$ and property \ref{item:thm:derivative-bound1} holds for some $\theta_i(x)$, then for any $k>\hat{k}~(k\in\N)$, 
there is a unique $\xi\in (k\pi -\pi/2,  k\pi+\pi/2)$ such that $\theta_i(\xi)=\tan{\xi}$.
\end{lemma}
\begin{proof}
For each $\theta_i'(x)$, we prove that one of the two cases holds.
Note that $\hat{M}$ is greater than the start-point of the terminal-root-functions of $g$ (recall Definition \ref{def:root-function-start-point}).
So, when $x\in(\hat{M},+\infty)$, $\lc(g,y)(x)\neq0$, $\discrim(g,y)(x)\ne0$ and $g(x,\theta_i(x))=0$. 
Therefore, $\frac{\partial g}{\partial y}(x,\theta_i(x))\ne0$ and $\theta_i'(x)=-\frac{\partial g/\partial x}{\partial g/\partial y}(x,\theta_i(x))$ by the implicit function theorem. 
By the definition of $\hat{M}$, we have 
$\lc(g,y)(x)\neq0$ and $\res(\frac{\partial g}{\partial x}+\delta_0 \frac{\partial g}{\partial y},g,y)(x)\ne0$ for $x\in(\hat{M},+\infty)$.
Then, since $g(x,\theta_i(x))=0$, it is clearly that
$(\frac{\partial g}{\partial x}+\delta_0\frac{\partial g}{\partial y})(x,\theta_i(x))\ne 0$.
So, we have $\theta_i'(x)\ne\delta_0$.
Therefore, by the continuity of $\theta'_{i}(x)$, exactly one case holds. 
    
As for the last assertion, the existence is by the intermediate value theorem applied to $\tan x-\theta_i(x)$ on the interval $(k\pi -\pi/2,  k\pi+\pi/2)$. The uniqueness follows from the fact that $\tan x-\theta_i(x)$ is strictly increasing on that interval since $1/\cos^2x-\theta'_i(x) >1-\delta_0>0$.
\end{proof}

Keep the notation in Lemma \ref{thm:derivative-bound}.
By the ``In particular'' part of Lemma \ref{thm:derivative-bound}, 
when $\theta_i'(x)<\delta_0$ on the interval $(\hat{M},\infty)$ for some selected $\delta_0<1$,
we can guarantee that there exists exactly $1$ periodic root (induced by $\theta_i$) of $g(x,\tan{x})$ in $(k\pi-\pi/2,k\pi+\pi/2)$ with $k$ sufficiently large. 
In order to deal with the case with $\theta_i'(x)>\delta_0$, we consider $1/\theta_i(x)$ instead of $\theta_i(x)$.
In fact, $1/\theta_i(x)$ is also a terminal-root-function of the reciprocal polynomial
$h:=\mathcal{R}(g)$ (see Def. \ref{definition:recip}). If $y\nmid g$, we have $\deg(h,y)=\deg(g,y)\geq1$ by the definition. Moreover, since $g$ is square-free and clearly $h\neq cy$ for any $c\in\Q$, $h$ is also square-free (recall Prop. \ref{proposition:recipKeepSquarefree}).

Finally, we obtain the following theorem, which strengthens Theorem \ref{thm:location-of-the-roots}.

\begin{theorem}\label{thm:Uniqueness}
Let $g\in\Q[x,y]$ be square-free such that $\deg(g,y)\geq1$ and $y\nmid g$. Set
$h=\mathcal{R}(g)$ and take
$\delta_1\in\Q ~(0<\delta_1<1)$, $\delta_2\in\Q ~(-1<\delta_2<0)$ such that neither $\res(\frac{\partial g}{\partial x}+\delta_1 \frac{\partial g}{\partial y},g,y)$ nor $\res(\frac{\partial h}{\partial x}+\delta_2 \frac{\partial h}{\partial y},h,y)$ is the zero polynomial.
Let $k'$ be the same as in Theorem \ref{thm:location-of-the-roots}, 
$k''$ a $\delta_1$-bound-point of $g$ and
$k'''$ a $\delta_2$-bound-point of $h$ guaranteed by Lemma \ref{lemma:existence-of-delta}.
Set $k^{*}:=\max\{k',k'',k'''\}$ and let $\theta_1(x),\ldots,\theta_r(x)$ be the terminal-root-functions of $g$.
Then, for any $k>k^{*}~(k\in\Z)$, 
there is a unique $\xi\in (k\pi -\pi/2,k\pi +\pi/2)$ such that $\theta_i(\xi)=\tan \xi$.
\end{theorem}
\begin{proof}
Since $y\nmid g$, $\deg(g,y)\ge 1$ and $g$ is square-free, we see that $\deg(h,y)\geq1$ and $h$ is square-free. 
Thus, $k'''$ is well-defined. 
Set $\mu:=(k^{*}+1/2)\pi$.
By Lemma \ref{thm:derivative-bound}, suppose $\theta_1(x),\ldots,\theta_t(x)$ are in Case \ref{item:thm:derivative-bound1} (\emph{i.e.}, $\theta_i'(x)<\delta_1<1$ for $x\in (\mu,+\infty))$ and $\theta_{t+1}(x),\ldots,\theta_r(x)$ are in Case \ref{item:thm:derivative-bound2} (\emph{i.e.}, $\theta_i'(x)>\delta_1$ for $x\in (\mu,+\infty))$ with $0\le t\le r$. 
Then, by the ``In particular'' part in Lemma \ref{thm:derivative-bound}, for each $\theta_i~(1\le i\le t)$ in Case \ref{item:thm:derivative-bound1}, we already have the conclusion. 
In the following, we deal with those $\theta_i(x)$ with $t+1\leq i\leq r$.

First, we prove that for each $\theta_i(x)~(t+1\le i\le r)$, $\theta_i(x)>0$ on $(\mu,+\infty)$: Since $\theta_i'(x)$ is greater than a positive number $\delta_1$ by Lemma \ref{thm:derivative-bound}, $\theta_i(x)\rightarrow+\infty$ as $x\rightarrow+\infty$. By property \ref{item:thm:location-of-the-roots1} of Theorem \ref{thm:location-of-the-roots}, the image of $\theta_i$ on $(\mu,+\infty)$ is a subset of $(a_{s+1},+\infty)$ with $a_{s+1}>0$ therein (recall Definition \ref{def:potential-periodic-interval}).

Second, we prove the existence of $\xi$ for each $\theta_i(x)~(t+1\le i\le r)$.
Let $k$ be any integer greater than $k^{*}$.
Applying the intermediate value theorem to $\tan{x}-\theta_i(x)$ on $(k\pi-\pi/2,k\pi+\pi/2)$, we see that there exists $\xi\in(k\pi-\pi/2,k\pi+\pi/2)$ such that $\tan{\xi}=\theta_i(\xi)$.

Finally, we prove the uniqueness of $\xi$ for each $\theta_i(x)~(t+1\le i\le r)$: Suppose $\tan\xi=\theta_i(\xi)$ for some $\xi\in(k\pi-\pi/2,k\pi+\pi/2)$. Then, $\xi\in(k\pi,k\pi+\pi/2)$ since $\theta_i(x)>0$ for all $x\in(\mu,+\infty)$.
Note that $1/\theta_i(x)$ is well defined on $(\mu,+\infty)$, $h(x,1/\theta_i(x))=0$ whenever $x\in (\mu,+\infty)$, and $\mu$ is greater than the start-point of the terminal-root-functions of $h$. 
So, $1/\theta_i(x)$ is one of the terminal-root-functions of $h$ when $x\in (\mu,+\infty)$.
By Lemma \ref{thm:derivative-bound}, we have 
$(1/\theta_i(x))'<\delta_2$ or $(1/\theta_i(x))'>\delta_2$ for $x\in(\mu,+\infty)$ with some $-1<\delta_2<0$.
Since $\lim\limits_{x\rightarrow +\infty}\theta_i(x)=+\infty$ and $\lim\limits_{x\rightarrow +\infty}1/\theta_i(x)=0$, we have $(1/\theta_i(x))'>\delta_2$. 
Because $-1/\sin^2 x-(1/\theta_i(x))'<-1-\delta_2<0$, 
the function $\cot x-1/\theta_i(x)$ is strictly decreasing on the interval $(k\pi,k\pi+\pi/2)$.
Thus, $\xi$ is unique.
\end{proof}

\begin{example}
Consider the polynomial $g(x,y)$ in Example \ref{example:potential-periodic-interval-set}.
In Example \ref{example:thm:location-of-the-roots}, we have taken $k'=3$.
It remains to compute $k''$ and $k'''$ in Theorem \ref{thm:Uniqueness}.
We take
$\delta_1=\frac{999}{1000}$ and $k''=3$, which is a 
$\delta_1$-bound-point of $g$.
Note that $h=\mathcal{R}(g)=10xy^8-20y^7-40xy^6+92y^5+60xy^4-28y^3-40xy^2+20y+10x.$
We take
$\delta_2=-\frac{999}{1000}$ and $k'''=0$, which is a $\delta_2$-bound-point of $h$.
Thus, we have $k^{*}=\max\{k',k'',k'''\}=3$.
By Example \ref{example:thm:location-of-the-roots} and Theorem \ref{thm:Uniqueness}, we see that for every $k > 3~(k \in \Z)$, there are exactly two real roots of $g(x,\tan{x})$ in the interval $(k\pi-\arctan{\frac{215}{128}},k\pi-\arctan{\frac{19}{32}})$.
\end{example}

We propose Algorithm \ref{alg:AlgDerBound} to compute 
a better lower bound, which guarantees the uniqueness, for the ``periodic'' roots of $g(x,\tan{x})$.
Given a square-free polynomial $g\in\Q[x,y]$ such that $\deg(g,y)\geq1$ and $y\nmid g$, together with a general potential periodic interval set $I$, 
Algorithm \ref{alg:AlgDerBound} computes a natural number $k^{*}$ and a set of pairs $((x_{i}^{-},x_{i}^{+}),c_i)$ such that for any $k>k^{*}~(k\in\N)$, there are exactly $c_i$ real roots of $g(x,\tan x)$ in the interval $(x_{i}^{-}+k\pi,x_{i}^{+}+k\pi)$.
The process of Algorithm \ref{alg:AlgDerBound} is as follows.
First, by Algorithm \ref{alg:AlgFarZero}, we compute $k'\in\N$ 
and a set of isolating intervals $ret$ such that for any $k>k'~(k\in\N)$ and any pair $((x_i^{-},x_{i}^{+}),c_i)$ in $ret$, there are at least $c_i$ real roots of $g(x,\tan x)$ in  $(x_i^{-}+k\pi,x_i^{+}+k\pi)$.
Then, we compute $k''\in\N$ and  $k'''\in\N$ by Lemma \ref{thm:derivative-bound} in Lines \ref{line:alg2-k''-st}--\ref{line:alg2-k''-ed} and Lines \ref{line:alg2-k'''-st}--\ref{line:alg2-k'''-ed}, respectively.
Finally, the algorithm returns the natural number $k^{*}:=\max\{k',k'',k'''\}$ and the set $ret$.
The correctness of Algorithm \ref{alg:AlgDerBound} is guaranteed by Theorem \ref{thm:Uniqueness} and the termination of it is obvious.

\begin{algorithm}[ht]\label{exactroot}
\scriptsize
\DontPrintSemicolon
\LinesNumbered
\SetKwInOut{Input}{Input}
\SetKwInOut{Output}{Output}
\Input{ 
$g$, a square-free polynomial in $\Q[x,y]$ such that $\deg(g,y)\geq 1$ and $y\nmid g$;\\
$I=[(a_0,b_0),\ldots,(a_{s+1},b_{s+1})]$, a  general potential periodic interval set of $g$.
}
\Output{
$k^{*}$, a natural number;\\
$ret$, a finite set of pairs $((x_i^{-},x_i^{+}),c_i)$ such that 
for any $k>k^{*}~(k\in\N)$, there are \emph{EXACTLY} $c_i$ roots for $g(x,\tan x)$ in the interval $(x_i^{-}+k\pi,x_i^{+}+k\pi)$.
}
\caption{\bf CompleteIsolatingFarZero}\label{alg:AlgDerBound}
\BlankLine

$k',\;ret\leftarrow{\bf IsolatingFarZero}(g,I)$\label{line:alg2-k'}\;
$d_1\leftarrow \mathtt{res}(\frac{\partial g}{\partial x}+\delta \frac{\partial g}{\partial y},g,y)$ \#\emph{In this step, $\delta$ is a variable and $d_1\in\Q[\delta,x]$.}\label{line:alg2-k''-st}\; 
Choose $\delta_1 \in(0,1)\cap \mathbb{Q}$ such that $d_1(\delta_1,x)$ is not the zero polynomial in $\Q[x]$.\;
Find $k''\in\N$ such that
$(k''+\frac{1}{2})\pi$ is an upper bound for all real roots of the polynomial $\textcolor{white}{hhhhhhhhhhhh}d_1(\delta_1,x)\cdot\lc(g,y)\cdot\allowbreak\discrim(g,y)$.\label{line:alg2-k''-ed}\;
$h\leftarrow g(x,\frac{1}{y})\cdot y^{\deg(g,y)}$\label{line:alg2-k'''-st}\;
$d_2\leftarrow \mathtt{res}(\frac{\partial h}{\partial x}+\delta \frac{\partial h}{\partial y},h,y)$\;
Choose $\delta_2 \in(-1,0)\cap \mathbb{Q}$ such that $d_2(\delta_2,x)$ is not the zero polynomial in $\Q[x]$.\;
Find $k'''\in\N$ such that
$(k'''+\frac{1}{2})\pi$ is an upper bound for all real roots of the polynomial  $\textcolor{white}{hhhhhhhhhhhh}d_2(\delta_2,x)\cdot\lc(h,y)\cdot\allowbreak\discrim(h,y)$.\label{line:alg2-k'''-ed}\;
$k^{*}\leftarrow \max\{k',k'',k'''\}$\;
\Return $k^{*}$, $ret$
\end{algorithm}

\section{Isolating Bounded Real Roots}\label{sec:IBR}
In this section, we discuss case \ref{item:MainThreeCase1} in Section \ref{sec:Preprocessing}, i.e., given $k_+,k_-\in\Z$, isolating the real roots of
\begin{align}\label{eq:boudedeq}
    g(x,\tan x)=0,~~~~x\in(k\pi-\frac{\pi}{2},k\pi+\frac{\pi}{2})
\end{align}
where $g\in\Q[x,y]\setminus\{0\}$, $k\in\Z$ and $k_-\le k\le k_+$.

We solve \eqref{eq:boudedeq} in four intervals $(k\pi-\frac{\pi}{2},k\pi-\frac{\pi}{4}),(k\pi-\frac{\pi}{4},k\pi),(k\pi,k\pi+\frac{\pi}{4})$ and $(k\pi+\frac{\pi}{4},k\pi+\frac{\pi}{2})$.
In each interval, the problem is equivalent to isolating the roots of some $G(r\pi+\arctan{y},y)$ on $(0,1)$ where $G\in\Q[x,y]$ and $r\in\Q$, and thus we can use the frame of \cite[Algorithm 47]{strzebonski2012real} to deal with it.
In Section \ref{subsec:WeakFourierSeq}, we use Algorithm \ref{alg:WeakFourierSeq} to compute a weak Fourier sequence for $G(r\pi+\arctan{y},y)$.
In fact, Algorithm \ref{alg:WeakFourierSeq} is a simplified version of \cite[Algorithm 40]{strzebonski2012real} for our specific needs.
In Section \ref{subsec:IsoRoot}, Algorithm \ref{alg:BoundRSI} is proposed to isolate the real roots of Eq. \eqref{eq:boudedeq}.

\subsection{Weak Fourier sequences}\label{subsec:WeakFourierSeq}

Let us first review the definition of weak Fourier sequences and an important result about it.

\begin{definition}\cite[Def. 10]{strzebonski2012real}\label{def:WeakFourierSeq}
Let $I\subseteq\R$ be an open interval, and $G_1,\ldots,G_k:I\rightarrow\R$ be a sequence of differentiable functions defined on $I$.
The sequence $G_1,\ldots,G_k$ is a \emph{weak Fourier sequence on $I$}, if for all $y\in I$,
\begin{enumerate}
    \item for all $i\;(1\le i\le k-1)$, $\sign(G_{i+1}(y))=\sign(G_i'(y))$, and
    \item $G_k(y)\ne0$.
\end{enumerate}
We say $G_1,\ldots,G_k$ is a \emph{weak Fourier sequence for $G$ on $I$}, if $G_1,\ldots,G_k$ is a weak Fourier sequence on $I$ and $G_1(y)=G(y)$ for all $y\in I$.
\end{definition}

\begin{notation}
Let $a,b\in\R$ and $G_1,\ldots,G_k$ be a weak Fourier sequence on the open interval $(a,b)$.
Then, for any $i,j\;(1\le i\le j\le k)$,
\begin{enumerate}
    \item $\sign(G_{i}(x^{+})):=\sign(G_i(c))$, where $x\in[a,b)$, $c< b$ and $\forall y\in(x,c]\;G_i(y)\ne0$,
    \item $\sign(G_{i}(x^{-})):=\sign(G_i(d))$, where $x\in(a,b]$, $a< d$ and $\forall y\in[d,x)\;G_i(y)\ne0$,
    \item $\sgc_{i,j}(x)$, for $x\in(a,b)$, is the number of sign changes in the sequence $G_{i}(x),\ldots,G_{j}(x)$ with terms equal to zero removed,    
    \item $\sgc_{i,j}(x^{+})$, for $x\in[a,b)$, is the number of sign changes in the sequence $G_{i}(x^{+}),\ldots,G_{j}(x^{+})$,
    \item $\sgc_{i,j}(x^{-})$, for $x\in(a,b]$, is the number of sign changes in the sequence $G_{i}(x^{-}),\ldots,G_{j}(x^{-})$.
\end{enumerate}
\end{notation}

\begin{theorem}\cite[Thm. 13]{strzebonski2012real}\label{thm:WeakFourierSeq}
Let $I\subseteq\R$ be an open interval and let 
$G_1,\ldots,G_{k}:I\rightarrow\R$ be a weak Fourier sequence.
Then for any $x\in I$, $\sgc_{1,k}(x^{+})=\sgc_{1,k}(x)$ and $\sgc_{1,k}(x^{-})=\sgc_{1,k}(x)+r+2s$, where $r,s\in\N$ and $G_1(x)=\cdots=G_r(x)=0,G_{r+1}(x)\ne0$.
Moreover, $s=0$ unless there is a $t>r+1$ such that $G_{t}(x)=0$.\qed
\end{theorem}

Let $a,b$ be two real numbers such that $(a,b)\subseteq I$.
Remark that by the proof of \cite[Prop. 15]{strzebonski2012real}, if $\sgc_{1,k}(a^{+})=\sgc_{1,k}(b^{-})$, then $G_1$ has no real root in $(a,b)$, and if $\sgc_{1,k}(a^{+})=\sgc_{1,k}(b^{-})+1$, then $G_1$ has exactly one simple root in $(a,b)$.

To illustrate Algorithm \ref{alg:WeakFourierSeq}, we introduce the definition of the \emph{arctan-derivative}.

\begin{definition}\label{def:pseudo}
For any $G\in\Q(y)[x]$, let
\[
{\trueder G}:=\frac{\partial G}{\partial x}\cdot\frac{1}{1+y^2}+\frac{\partial G}{\partial y}\in\Q(y)[x].
\]
The \emph{arctan-derivative} of $G$ is defined as $G^{\pder}:={\trueder G}\cdot h(y)$, where $h(y)\in\Q[y]$ is the denominator of $\lc({\trueder G},x)\in\Q(y)$.
In particular, if $\lc({\trueder G},x)$ is a polynomial, then $h(y)=1$.
We call $h(y)$ the \emph{common-term} of $G^{\pder}$.
\end{definition}

\begin{remark}\label{remark:def:pseudo}
For any $G\in\Q(y)[x]$, $G^{\pder}\in\Q(y)[x]$ and $\lc(G^{\pder},x)$ is a polynomial in $\Q[y]$.
\end{remark}

Given $G(x,y)\in\Q[x,y]\setminus\{0\}$ and $r\in\Q$,
the process of Algorithm \ref{alg:WeakFourierSeq} is as follows:
Let $G_1$ be $G(x,y)$.
If $G_1$ is not a nonzero rational number, then compute the arctan-derivative $G_{1}^{\pder}$ of $G_1$, and set $G_2=G_{1}^{\pder}$.
If $G_2$ is not a nonzero rational number, then set $G_3=G_{2}^{\pder}$.
Repeat the process until some $G_k\;(k\ge1)$ is a nonzero rational number.
Then $G_1(r\pi+\arctan y,y),\allowbreak\ldots,G_k(r\pi+\arctan y,y)$ is a weak Fourier sequence for $G(r\pi+\arctan y,y)$ on $\R$.

\begin{algorithm}[!ht]
\scriptsize
\DontPrintSemicolon
\LinesNumbered
\SetKwInOut{Input}{Input}
\SetKwInOut{Output}{Output}
\Input{ $G(x,y)$, a polynomial in $\Q[x,y]\setminus\{0\}$; $r$, a rational number
}
\Output{$G_1(r\pi+\arctan{y},y),\ldots,G_k(r\pi+\arctan{y},y)$, a weak Fourier sequence for $G(r\pi+\arctan{y},y)$ on $\R$, where $G_i\in\Q(y)[x]$
}
\caption{{\bf WeakFourierSeq}}\label{alg:WeakFourierSeq}
\BlankLine
$G_1\leftarrow G(x,y)$\;
$i\leftarrow1$\;
\While{$G_i(x,y)$ is not a nonzero rational number}
{ $G_{i+1}(x,y)\leftarrow G_{i}^{\pder}(x,y)$\;
  $i\leftarrow i+1$
}
\Return{$G_1(r\pi+\arctan{y},y),\ldots,G_k(r\pi+\arctan{y},y)$}
\end{algorithm}

Note that we can easily obtain a weak Fourier sequence from a \emph{semi-Fourier sequence} (defined in \cite[Def. 14]{strzebonski2012real}) by \cite[Prop. 15]{strzebonski2012real}.
So, Algorithm \ref{alg:WeakFourierSeq} and \cite[Algorithm 40]{strzebonski2012real} are essentially the same when computing a weak Fourier sequence for $G(r\pi+\arctan{y},y)$, and thus the termination and correctness of Algorithm \ref{alg:WeakFourierSeq} are clear.
However, in order to explain some properties that will be used later, we prepare some lemmas and re-prove the termination and correctness of Algorithm \ref{alg:WeakFourierSeq}.

\begin{lemma}\label{lemma:DerEquiv0}
If $G\in\Q(y)[x]\setminus\Q$ such that $G(x,y)$ is well-defined on the whole plane $\R^2$,
then $G^{\pder}\in\Q(y)[x]\setminus\{0\}$.
\end{lemma}
\begin{proof}
Assume that $G^{\pder}$ is the zero polynomial in $\Q(y)[x]$.
By Definition \ref{def:pseudo}, ${\trueder G}$ is also the zero polynomial.
Then, ${\trueder G}(\arctan{y},y)\equiv0$.
Note that ${\trueder G}(\arctan{y},y)=\frac{\der G(\arctan{y},y)}{\der y}$.
So, $G(\arctan{y},y)\in\Q$ for any $y\in\R$. Since the arctangent function is not algebraic, $G\in\Q$, a contradiction.
\end{proof}

\begin{lemma}\label{lemma:pseudoproperty}
If $G\in\Q(y)[x]\setminus\Q$ and $\lc(G,x)$ is a polynomial, then
\begin{align}\label{eq:lemma:pseudoproperty}
    (\deg(G,x),\deg(\lc(G,x)))\succlex(\deg(G^{\pder},x),\deg(\lc(G^{\pder},x))),
\end{align}
where $\succlex$ denotes the lexicographic ordering.
\end{lemma}
\begin{proof}
Recall Remark \ref{remark:def:pseudo}. The degrees $\deg(G^{\pder},x)$ and $\deg(\lc(G^{\pder},x))$ are well-defined.

If $\deg(G,x)=0$, we have $G\in\Q(y)\setminus\Q$, and thus $\lc(G,x)=G$.
So, $G$ is a polynomial in $\Q[y]\setminus\Q$.
It is easy to check that $G^{\pder}=\frac{\partial G}{\partial y}\in\Q[y]$. 
So, $\deg(\lc(G,x))=\deg(G)$, $\deg(G^{\pder},x)=0$ and $\deg(\lc(G^{\pder},x))=\deg(G)-1$. 
Then inequality \eqref{eq:lemma:pseudoproperty} holds.

If $\deg(G,x)\ge1$, suppose that 
\begin{align}\label{eq:G}
G=a_n(y)x^n+\sum_{j=0}^{n-1}a_j(y)x^j
\end{align}
where $n=\deg(G,x)\ge1$, $a_n(y)\in\Q[y]$ and $a_j(y)\in\Q(y)$ for $j=0,\ldots,n-1$.
Then,
\begin{align}\label{eq:truederG}
    {\trueder G}=a_n'(y)x^n+&\sum_{j=0}^{n-1}\big(a_j'(y)+\frac{(j+1)a_{j+1}(y)}{1+y^2}\big)x^{j}.
\end{align}
If $a_{n}'(y)$ is the zero polynomial, we have $\deg({\trueder G},x)<\deg(G,x)$.
By Definition \ref{def:pseudo}, $\deg(G^{\pder},x)=\deg({\trueder G},x)$.
So, $\deg(G^{\pder},x)<\deg(G,x)$ and inequality \eqref{eq:lemma:pseudoproperty} is true.
If $a_{n}'(y)$ is not the zero polynomial, then $\deg({\trueder G},x)=\deg(G,x)$.
Since $\lc({\trueder G},x)=a_{n}'(y)\in\Q[y]$, $G^{\pder}={\trueder G}$. Thus $\deg(G^{\pder},x)=\deg(G,x)$.
We also note that $\deg(\lc(G,x))=\deg(a_n(y))$ and $\deg(\lc(G^{\pder},x))=\deg(a_n(y))-1$.
The proof is completed.
\end{proof}

\begin{lemma}\label{lemma:DerForm}
If $G\in\Q(y)[x]$ has the following form
\begin{align}\label{eq:specialform}
    G=a_n(y)x^{n}+\sum_{j=0}^{n-1}\frac{a_j(y)}{(1+y^2)^{m_j}}x^{j},
\end{align}
where $n=\deg(G,x)$, $a_i(y)\in\Q[y]\;(0\le i\le n)$ and $m_j\in\N~(0\le j\le n-1)$, then the common-term of $G^{\pder}$ is $(1+y^2)^m$ for some $m\in\N$.
Furthermore, $G^{\pder}$ also has the form in Eq. \eqref{eq:specialform}.
\end{lemma}
\begin{proof}
If $\deg(G,x)=0$, then $G=a_n(y)\in\Q[y]$, and thus $G^{\star}=G'(y)$, the common-term of $G^{\star}$ is $1$.
If $\deg(G,x)\ge1$, by Eq. \eqref{eq:G} and Eq. \eqref{eq:truederG} in the proof of Lemma \ref{lemma:pseudoproperty}, there exist $b_i(y)\in\Q[y]\;(0\le i\le n)$ and $d_j\in\N~(0\le j\le n-1)$ such that
\begin{align*}
    {\trueder G}=b_n(y)x^n+\sum_{j=0}^{n-1}\frac{b_j(y)}{(1+y^2)^{d_j}}x^{j}.
\end{align*}
If $b_n(y)$ is not the zero polynomial, the common-term of $G^{\pder}$ is $1$.
Otherwise, since $\deg(G,x)\ge1$ and $G$ is well-defined on the whole $\R^2$, by the proof of Lemma \ref{lemma:DerEquiv0}, ${\trueder G}$ is not the zero polynomial.
So, there exists an integer $0\le k\le n-1$ such that $b_k(y)$ is not the zero polynomial. Let $k_0$ be the maximal one of such $k$. Then, the common-term of $G^{\pder}$ is $(1+y^2)^{d_{k_0}}$.
Therefore, $G^{\pder}$ has the same form as in Eq. \eqref{eq:specialform}.
\end{proof}

\begin{theorem}\label{thm:alg:WeakFourierSeqTC}
   Algorithm \ref{alg:WeakFourierSeq} terminates correctly.
\end{theorem}
\begin{proof}
(Termination) 
For any $G(x,y)\in\Q[x,y]\setminus\{0\}$,
let $G_1:=G(x,y),G_2:=G_1^{\pder}(x,y),G_3:=G_2^{\pder}(x,y),\ldots$
We only need to prove that there exists $k\;(k\ge1)$ such that $G_k(x,y)$ is a nonzero rational number.
If $G_1\in\Q\setminus\{0\}$, then we take $k=1$.
If $G_1\in\Q[x,y]\setminus\Q$, then by Lemma \ref{lemma:pseudoproperty}, we have 
$$(\deg(G_1,x),\deg(\lc(G_1,x)))\succlex(\deg(G_2,x),\deg(\lc(G_2,x))).$$
Recall Remark \ref{remark:def:pseudo}. For any $i~(i\ge2)$, $G_i\in\Q(y)[x]$ and $\lc(G_i,x)$ is a polynomial.
Therefore, assume that $G_2\notin\Q$, $G_3\notin\Q$,..., by Lemma \ref{lemma:pseudoproperty}, we have
\begin{align*}
(\deg(G_1,x),\deg(\lc(G_1,x)))&\succlex(\deg(G_2,x),\deg(\lc(G_2,x)))\\
&\succlex(\deg(G_3,x),\deg(\lc(G_3,x)))\succlex\cdots.
\end{align*}
So, there exists $m\;(m\ge2)$ such that $G_m\in\Q$.
Let $k\;(k\ge2)$ be the smallest integer such that $G_k\in\Q$.
Then, $G_{k-1}\notin\Q$.
Note that $G_{k-1}$ is well-defined on the whole $\R^2$ by Lemma \ref{lemma:DerForm}.
So, by Lemma \ref{lemma:DerEquiv0}, $G_k$ is a nonzero rational number.

(Correctness) 
For $i=1,\ldots,k-1$, let $h_i(y)$ be the common-term of $G_i^{\pder}(x,y)$.
Recall Definition \ref{def:pseudo}. 
We have $G_{i+1}(x,y)=h_i(y){\trueder G}_{i}(x,y)$.
Note that 
$${\trueder G}_{i}(r\pi+\arctan{y},y)=G_i(r\pi+\arctan{y},y)'$$
and $G_{i+1}(x,y)\in\Q(y)[x]$.
Thus, we only need to prove that $\sign(h_i(y))>0$ on the whole $\R$ and $G_{i+1}(r\pi+\arctan{y},y)$ is defined on the whole $\R$.
If $\deg(G,x)=0$, then $G\in\Q[y]$.
So, $h_i(y)=1$ and $G_{i+1}(r\pi+\arctan{y},y)$ is defined on the whole $\R$.
If $\deg(G,x)\ge1$, 
by the termination proof, there exists $m\;(1\le m\le k-1)$ such that $\deg(G_1,x)\ge1,\ldots,\deg(G_m,x)\ge1$ but $\deg(G_{m+1},x)=0$.
Applying Lemma \ref{lemma:DerForm} to $G_i\;(i=1,\dots,m)$, we have 
$h_i(y)=(1+y^2)^{n_i}\;(n_i\in\N)$ and $G_{i+1}(r\pi+\arctan y,y)$ is defined on the whole $\R$.
Since $\deg(G_{m+1},x)=0$, $G_{m+1}\in\Q[y]$.
So, for any $m+1\le i\le k-1$, $h_i(x)=1$ and 
$G_{i+1}\in\Q[y]$, which is also defined on the whole $\R$.
\end{proof}

\begin{remark}\label{remark:PolyWeakFourierSeq}
Let $G(x,y)\in\Q[x,y]\setminus\{0\}$, $r\in\Q$ and $G_1(r\pi+\arctan{y},y),\ldots,\allowbreak G_k(r\pi+\arctan{y},y)$ be a weak Fourier sequence for $G(r\pi+\arctan{y},y)$ on $\R$ computed by Algorithm \ref{alg:WeakFourierSeq}.
By the correctness proof of Theorem \ref{thm:alg:WeakFourierSeqTC}, $G_i(x,y)\in\Q[y]$ or
\begin{align*}
    G_i(x,y)=a_{in_i}(y)x^{n_i}+\sum_{j=0}^{n_{i}-1}\frac{a_{ij}(y)}{(1+y^2)^{m_{ij}}}x^{j},
\end{align*}
where $n_i=\deg(G_i,x)\ge1$, $a_{in_i}(y)\in\Q[y]~(1\le i\le k)$ and $m_{ij}\in\N~(0\le j\le n_{i}-1)$.
For $i=1,\ldots,k$, let $m_i:=\max(m_{ij},0\le j\le n_i-1)$ and
\begin{align*}
F_i:=
\left\{
\begin{array}{cc}
G_i, & G_i\in\Q[y],\\
(1+y^2)^{m_{i}}G_i,     &  {\rm otherwise}.
\end{array}
\right.
\end{align*}
We have $F_i\in\Q[x,y]$, and $F_i(r\pi+\arctan{y},y)$ has the same signs and real roots as $G_i(r\pi+\arctan{y},y)$ on $\R$.
Thus, we consider $F_i$ instead of $G_i$ in the next subsection.
\end{remark}

\subsection{Isolating the roots}\label{subsec:IsoRoot}
Recall that the goal of Section \ref{sec:IBR} is to isolate the roots of the following equation: 
\begin{align}\label{meq:boudedeq1_9}
    g(x,\tan{x}),~x\in(k\pi-\frac{\pi}{2},k\pi+\frac{\pi}{2}),~{\rm where}~k\in\Z.
\end{align}

\subsubsection{The algorithm}
Given $g(x,y)\in\Q[x,y]\setminus\{0\}$ and $k\in\Z$,
we propose Algorithm \ref{alg:BoundRSI} to isolate the real roots of Eq. \eqref{meq:boudedeq1_9}.

The process of Algorithm \ref{alg:BoundRSI} is as follows.
We first decide whether $k\pi$, $-\frac{\pi}{4}+k\pi$ or $\frac{\pi}{4}+k\pi$ is a root of $g(x,\tan{x})$ by Lemma \ref{lemma:piRealRoot}.
Then, we compute all irreducible factors $g_1,\ldots,g_m$ of $g(x,y)$ in $\Q[x,y]$.
For every $g_i(x,\tan{x})$, we only need to isolate the real roots of it in four open intervals, i.e., $(k\pi-\frac{\pi}{2},k\pi-\frac{\pi}{4})$, $(k\pi-\frac{\pi}{4},k\pi)$, $(k\pi,k\pi+\frac{\pi}{4})$ and $(k\pi+\frac{\pi}{4},k\pi+\frac{\pi}{2})$.
If $g_i$ is a polynomial in $\Q[x]$, it is easy to deal with (there are many methods to isolate the real roots of univariate polynomials).
Otherwise, we have $\deg(g_i,y)\ge1$. 
Let $g_{i}^{(1)}(x,y):=g_{i}(-x,-y)$.
It is easy to check that isolating the real roots of $g_i(x,\tan{x})$ in the first interval (\emph{resp.}, the second interval) is equivalent to isolating those of $g_{i}^{(1)}(x,\tan{x})$ in $(-k\pi+\frac{\pi}{4},-k\pi+\frac{\pi}{2})$ (\emph{resp.}, $(-k\pi,-k\pi+\frac{\pi}{4})$).
Note that these four cases can all be reduced to the problem of isolating the real roots of $G(x,\tan{x})$ on $(r\pi,r\pi+\frac{\pi}{4})$, where $G\in\{g_i,g_{i}^{(1)}\}$ is a rational bivariate irreducible polynomial with $\deg(G,y)\ge1$ and $r\in\{k_0,k_0+\frac{\pi}{4}\}~(k_0\in\Z)$.  
And, Algorithm \ref{alg:SubIBR} can solve such problem, which we will explain in detail in the next subsection. Finally, in Algorithm \ref{alg:AlgDerBound} the isolating intervals obtained from each irreducible factor $g_i$ are refined so that they do not intersect. This is by using Remark \ref{remark:bisect} and the fact that the roots in them are all distinct simple roots (Prop. \ref{prop:multiplicity-one}--\ref{proposition:nocommonroot}).

\begin{algorithm}[!ht]
\scriptsize
\DontPrintSemicolon
\LinesNumbered
\SetKwInOut{Input}{Input}
\SetKwInOut{Output}{Output}
\Input{
$g(x,y)$, a polynomial in $\Q[x,y]\setminus\{0\}$\\
$k$, an integer
}
\Output{$ans$, a set of pairwise disjoint isolating intervals of real roots of $g(x,\tan{x})$ in $(k\pi-\frac{\pi}{2},k\pi+\frac{\pi}{2})$ 
}
\caption{{\bf IsolatingBoundedRoots}}\label{alg:BoundRSI}
\BlankLine
$ans\leftarrow\emptyset$\;
If $g(k\pi,0)=0$, then $ans\leftarrow ans\cup\{[k\pi,k\pi]\}$\;
If $g(-\frac{\pi}{4}+k\pi,-1)=0$, then $ans\leftarrow ans\cup\{[-\frac{\pi}{4}+k\pi,-\frac{\pi}{4}+k\pi]\}$\;
If $g(\frac{\pi}{4}+k\pi,1)=0$, then $ans\leftarrow ans\cup\{[\frac{\pi}{4}+k\pi,\frac{\pi}{4}+k\pi]\}$\;
$g_1,\ldots,g_m\leftarrow$ all irreducible factors of $g(x,y)$ in $\Q[x,y]$\;
\For{$i$ from $1$ to $m$}
{
\eIf{$g_i$ is a polynomial in $\Q[x]$}
{$S\leftarrow$ a set of isolating intervals of real roots of $g_i(x)$ in $(k\pi-\frac{\pi}{2},k\pi+\frac{\pi}{2})\setminus\{k\pi,-\frac{\pi}{4}+k\pi,\frac{\pi}{4}+k\pi\}$\;
$ans\leftarrow ans\cup S$\;
}
{
let $g_{i}^{(1)}\in\Q[x,y]$ such that $g_{i}^{(1)}(x,y)=g_{i}(-x,-y)$\;
$S_1\leftarrow{\bf SubIBR}(g_{i}^{(1)},-k+\frac{1}{4})$\;
$S_2\leftarrow{\bf SubIBR}(g_{i}^{(1)},-k)$\;
$S_3\leftarrow{\bf SubIBR}(g_{i},k)$\;
$S_4\leftarrow{\bf SubIBR}(g_{i},k+\frac{1}{4})$\;
$ans\leftarrow ans\cup\{(-b,-a)\mid (a,b)\in S_1\cup S_2\}\cup S_3\cup S_4$\;
}
}
Refine the intervals obtained in the loop above such that they do not intersect.\;
\Return{$ans$}\;
\end{algorithm}

\subsubsection{Sub-algorithms}
Let $g(x,y)$ be an irreducible polynomial in $\Q[x,y]$ with $\deg(g,y)\ge1$ and $r\in\{k,k+\frac{1}{4}\}$ where $k\in\Z$. 
According to the frame of \cite[Algorithm 47]{strzebonski2012real}, Algorithm \ref{alg:SubIBR} is proposed to isolate the real roots of $g(x,\tan{x})$ in $(r\pi,r\pi+\frac{\pi}{4})$.
It is worth noting that \cite[Algorithm 47]{strzebonski2012real} and Algorithm \ref{alg:SubIBR} are not the same.
Some steps are added, and some steps are skipped or implemented in another way.

The algorithm \cite[Algorithm 47]{strzebonski2012real} can be used to isolate the real roots of an exp-log-arctan function $w(y)$ in an open interval $(a,b)$, where $a,b$ are rational numbers and $(a,b)$ are contained in the domain of $w(y)$.
Let us review the frame of the algorithm.
We need to compute a weak Fourier sequence $w_1,\ldots,w_k$ for $w(y)$ on $(a,b)$ first.
Let $ans$ be a set of computed isolating intervals of $w(y)$ (initialized as $\emptyset$), and $stack$ be a stack of triples $(c,d,m)$ such that $c<d$ are rational numbers and $w_m$ has a constant nonzero sign in $(c,d)$ (initialized as $(a,b,k)$).
In each loop step, we pick a triple $(c,d,m)$ from $stack$ and compute the difference value $\sgc_{1,m}(c^{+})-\sgc_{1,m}(d^{-})$ by \cite[Algorithm 44]{strzebonski2012real}, 
until $stack$ is empty.
\begin{enumerate}
    \item If the value equals $0$, by Theorem \ref{thm:WeakFourierSeq}, $w(y)$ has no roots on $(c,d)$. So, we continue for the next loop step. 
    \item If the value equals $1$, by Theorem \ref{thm:WeakFourierSeq}, $w(y)$ has exactly one simple root on $(c,d)$. We add $(c,d)$ to $ans$.
    \item Otherwise, let $\ell$ be the smallest element of $\{i\mid 2\le i\le m \;{\rm and}\; \sgc_{i,m}(c^{+})\allowbreak-\sgc_{i,m}(d^{-})=1\}$.
    By Theorem \ref{thm:WeakFourierSeq}, $w_{\ell}(y)$ has exactly one simple root $sr$ on $(c,d)$.
    \begin{enumerate}
        \item If $sr$ is a root of $w(y)$, find two rational numbers $u,v$ such that $c<u<sr<v<d$ and $w(y)$ has only one root $sr$ in $[u,v]$
    (recall that by Lemma \ref{lemma:NoLimit}, we can find such $u,v$).
    Add $(u,v)$ to $ans$, and add $(c,u,\ell)$, $(v,d,\ell)$ to $stack$.
    \item If $sr$ is not a root of $w(y)$. Find two rational numbers $u,v$ such that $c<u<sr<v<d$ and $w(y)$ has no roots in $[u,v]$.     
    Add $(c,u,\ell)$ and $(v,d,\ell)$ to $stack$.
    \end{enumerate}
\end{enumerate}

Before we explain Algorithm \ref{alg:SubIBR}, we first define a map $\varphi$ and introduce Lemma \ref{lemma:mappingirre} and Lemma \ref{lemma:NoCommonRoot}.
The map $\varphi:\;{\mathbb Q}[x,y]\to{\mathbb Q}[x,y_1]$  is defined as
\begin{align}\label{eq:def-map1}
\varphi(g(x,y))=g(x,\frac{1+y_1}{1-y_1})\cdot(1-y_1)^{\deg(g,y)}.
\end{align}

\begin{lemma}\label{lemma:mappingirre}
If $g$ is an irreducible polynomial in $\Q[x,y]\setminus\{c(y+1)\mid c\ne0\}$, then $\varphi(g)$ is an irreducible polynomial in $\Q[x,y_1]$.
\end{lemma}
\begin{proof}
Since $\deg(g(x,2y-1),y)=\deg(g(x,y),y)$, $\varphi(g)=\mathcal{R}(g(x,2y-1))|_{y=1-y_1}$ with $\mathcal{R}$ in Definition \ref{definition:recip}. Now that $g(x,y)\neq c(y+1)$ for any nonzero $c\in\Q$, $g(x,2y-1)\neq cy$ for any nonzero $c\in\Q$. Moreover, $g(x,2y-1)$ is irreducible since $g(x,y)$ is. By Proposition \ref{proposition:RkeepIrr}, $\mathcal{R}(g(x,2y-1))$ is irreducible. Thus so is $\varphi(g)=\mathcal{R}(g(x,2y-1))|_{y=1-y_1}$.
\end{proof}
The following lemma is similar to \cite[Lemma 3.2]{chen2022automated} and \cite[Theorem 7]{mccallum2012deciding}.
\begin{lemma}\label{lemma:NoCommonRoot}
Let $F_1,F_2\in\Q[x,y]$, $k\in\Z$ and $r\in\{k,k+\frac{1}{4}\}$.
If ${\gcd}(F_1,F_2)=1$, then $F_1(r\pi+\arctan{y},y)$ and $F_2(r\pi+\arctan{y},y)$ have no common real roots in $(0,1)$.
\end{lemma}
\begin{proof}
Replacing those ``$p,q$" in the proof of \cite[Theorem 7]{mccallum2012deciding} by $F_1$ and $F_2$ in this lemma, those ``$f^*(x),g^*(x)$" therein by $F_1(r\pi+\arctan{y},y)$ and $F_2(r\pi+\arctan{y},y)$, and the function ``\text{trans}" therein by the function $r\pi+\arctan$, we see that: If $\beta\in(0,1)$ is a common root of $F_1(r\pi+\arctan{y},y)$ and $F_2(r\pi+\arctan{y},y)$, then both $\beta$ and $\alpha=r\pi+\arctan{\beta}$ would be algebraic, which contradicts Claim \ref{claim:atleastonetrans} below.
\end{proof}
\begin{claim}\label{claim:atleastonetrans}
For any $y\in(0,1)$, at least one of $r\pi+\arctan{y}$ and $y$ is transcendental. 
\end{claim}
\begin{proof}
Note that \[
\tan(r\pi+\arctan y)=\left\{
\begin{array}{ccl}\vspace{2mm}
     y,& \text{ if } &r=k\in\mathbb{Z},\\
     \frac{1+y}{1-y},& \text{ if }  &r=k+\frac{1}{4}\in\mathbb{Z}+\frac{1}{4}.
\end{array}
\right.
\]
Suppose $y$ is algebraic, then $\tan(r\pi+\arctan y)$ is also algebraic in either case. Since neither $y$ nor $(1+y)/(1-y)$ is zero, $r\pi+\arctan y\neq0$. Moreover, it is transcendental by Theorem \ref{thm:BothAlgeNum}.
\end{proof}
\begin{remark}\label{remark:bisect}
Suppose $h(x_1,x_2)=\sum_{i=0}^{\deg(h,x_1)}c_i(x_2)x_1^i\in\mathbb{Q}[x_1,x_2]$. For any $y\in(0,1)\cap\mathbb{Q}$, $r\pi+\arctan y$ is transcendental. Thus $h(r\pi+\arctan y,y)=0$ iff $c_i(y)=0$ for all $i$. We can therefore decide whether $h(r\pi+\arctan y,y)$ is zero for given $y\in(0,1)\cap\mathbb{Q}$. If it is nonzero, we can further compute its sign by approximating it more accurately. More generally, if $\phi$ is a real function $($\emph{e.g.}, $\phi=\tan)$ such that $\phi(y)$ is transcendental for any nonzero rational $y$ in its domain, then we can decide whether $h(\phi(y),y)$ is zero, and compute its sign whenever it is not zero.
\end{remark}

Recall that our goal is to isolate the real roots of $g(x,\tan{x})$ in $(r\pi,r\pi+\frac{\pi}{4})$, where $g$ is an irreducible polynomial in $\Q[x,y]$ with $\deg(g,y)\ge1$ and $r\in\{k,k+\frac{1}{4}\}~(k\in\Z)$. 
The process of Algorithm \ref{alg:SubIBR} is as follows.
If $g=y+1$, then $\tan{x}+1$ has no roots on $(r\pi,r\pi+\frac{\pi}{4})$, and thus return $\emptyset$.
Otherwise, let $F_1:=g$ when $r=k$, and $F_1:=\varphi(g)$ when $r=k+\frac{1}{4}$.
It is not difficult to check that our goal is equivalent to solving $F_1(x+r\pi,\tan{x})$ on $(0,\frac{\pi}{4})$.
Let $y:=\tan{x}$.
We only need to solve $F_1(r\pi+\arctan{y},y)$ on $(0,1)$.
Recall Remark \ref{remark:PolyWeakFourierSeq}.
We can compute a weak Fourier sequence
\begin{align}\label{eq:ComputeFourierSeq}
    w_1=F_1(r\pi+\arctan{y},y),\ldots,w_k=F_k(r\pi+\arctan{y},y)~{\rm on}~(0,1),
\end{align}
where $F_i\in\Q[x,y]$.
Thus, the frame of \cite[Algorithm 47]{strzebonski2012real} can be used.
However, there are two main differences between \cite[Algorithm 47]{strzebonski2012real} and Algorithm \ref{alg:SubIBR}.

First, when $w_{\ell}$ has exactly one simple root $sr$ on $(c,d)$ (see Line \ref{algline:SubIBR-sr}),
we propose a new method to decide whether $sr$ is a root of $w_1$. 
If $F_{\ell}=QF_1~(Q\in\Q[y])$ holds and $Q$ has no roots on $(c,d)$, by \eqref{eq:ComputeFourierSeq}, $sr$ is the only root of $w_1$ on $(c,d)$.
If $F_{\ell}=QF_1$ holds and $Q$ has roots on $(c,d)$, then it is clear that $w_1$ has no roots on $(c,d)$.
If $F_{\ell}=QF_1$ does not hold, then $sr$ is not a root of $w_1$.
This method is correct: By Lemma \ref{lemma:mappingirre}, $F_1$ is an irreducible polynomial.
If $F_1\mid F_{\ell}$, we have $\deg(F_1,x)\le\deg(F_{\ell},x)$.
Then, by Lemma \ref{lemma:pseudoproperty}, $\deg(F_1,x)=\deg(F_{\ell},x)$.
So, $Q$ is a polynomial in $\Q[y]$.
Otherwise, ${\gcd}(F_1,F_{\ell})=1$. Then, by Lemma \ref{lemma:NoCommonRoot}, $sr$ is not a root of $w_1$.

Second, if $sr$ is not a root of $w_1$, we find two rational numbers $u,v$ such that $c\leq u<sr<v\leq d$ and $w_1$ has no roots in $[u,v]$ by Algorithm \ref{alg:RationalEndpoint} (see Line \ref{algline:SubIBR-RationalEP} of Algorithm \ref{alg:SubIBR}), instead of by \cite[Algorithm 45]{strzebonski2012real}. The process of Algorithm \ref{alg:RationalEndpoint} is as follows: Let $n:=1$.
While $n\ge1$, we compute an upper limit polynomial $W_{\max,n}\in\Q[y]$ and a lower limit polynomial $W_{\min,n}\in\Q[y]$ for $w_1(y)$ by two steps.
Since $w_1=F_1(r\pi+\arctan{y},y)$, we first estimate $\arctan{y}$ through Taylor series, and $T_{\max}(n,w_1),T_{\min}(n,w_1)\in\Q[\pi][y]$ are computed by \cite[Def. 3.1]{chen2020automated}.
Then, we use a bounded interval with rational endpoints to estimate $\pi$ in $T_{\max}(n,w_1)$ and $T_{\min}(n,w_1)$, and the polynomials $W_{\max,n},W_{\min,n}\in\Q[y]$ are computed by \cite[Def. 3.1]{xia2006real}.
Let $q:=\frac{c+d}{2}$.
If $W_{\max,n}$ has no roots on $[c,d]$ and $W_{\max,n}(q)<0$, or if $W_{\min,n}$ has no roots on $[c,d]$ and $W_{\min,n}(q)>0$, then $w_1$ has no roots on $[c,d]$.
Otherwise, increase $n$ by $1$.
Since $w_{\ell}$ has exactly one simple root $sr$ on $(c,d)$, we use a bisection method based on Remark \ref{remark:bisect} to find the half open interval where $sr$ is.
Update $(c,d)$ as the half interval and go for the next loop step.
Remark that if the interval $(c,d)$ is small enough and the upper (lower) limit polynomial is close enough, Algorithm \ref{alg:RationalEndpoint} terminates by Lemma \ref{lemma:NoLimit}.

\begin{algorithm}[!t]
\scriptsize
\DontPrintSemicolon
\LinesNumbered
\SetKwInOut{Input}{Input}
\SetKwInOut{Output}{Output}
\Input{
$g(x,y)$, an irreducible polynomial in $\Q[x,y]$ with $\deg(g,y)\ge1$\\
$r$, $r\in\{k,k+\frac{1}{4}\}$ where $k\in\Z$
}
\Output{$ans$, a set of isolating intervals of real roots of $g(x,\tan{x})$ in $(r\pi,r\pi+\frac{\pi}{4})$ 
}
\caption{{\bf SubIBR} 
}\label{alg:SubIBR}
\BlankLine
{\bf if} $g=y+1$, {\bf then} {\bf return} $\emptyset$\;
\eIf{$r=k$}
{$F_1\leftarrow g$}
{$F_1\leftarrow g(x,\frac{1+y}{1-y})\cdot(1-y)^{\deg(g,y)}$\;}
let $w_1=F_1(r\pi+\arctan{y},y),\ldots,w_k=F_k(r\pi+\arctan{y},y)$, where $F_i\in\Q[x,y]$, be the sequence in Remark \ref{remark:PolyWeakFourierSeq}\label{algline:SubIBR-weakFseq}\;
$ans\leftarrow\emptyset$\;
$stack\leftarrow\{(0,1,k)\}$\;
\While{$stack\ne\emptyset$}
{
$(c,d,m)\leftarrow First(stack),\;stack\leftarrow Rest(stack)$\;
$sgc\leftarrow\sgc_{1,m}(c^{+})-\sgc_{1,m}(d^{-})$\;
{\bf if} $sgc=0$, {\bf then} continue\;
\If{$sgc=1$}{$ans\leftarrow ans\cup\{(r\pi+\arctan{c},r\pi+\arctan{d})\}$\; continue\;}
$\ell\leftarrow\min\;\{i\mid 2\le i\le m,~\sgc_{i,m}(c^{+})-\sgc_{i,m}(d^{-})=1\}$\;
$sr\leftarrow\Root(w_{\ell},c,d)$\label{algline:SubIBR-sr}\;
By Theorem \ref{thm:WeakFourierSeq}, $w_{\ell}$ has exactly one simple root $sr$ in $(c,d)$.
We use the notation $\Root(w_{\ell},c,d)$ to represent the only simple root of $w_{\ell}$ in $(c,d)$.\;
\eIf{$F_{\ell}=QF_1$ where $Q\in\Q[y]$\label{algline:SubIBR-wheHold}}
{
\eIf{$Q$ has no real roots on $(c,d)$}
{$ans\leftarrow ans\cup\{(r\pi+\arctan{c},r\pi+\arctan{d})\}$}
{continue}
}
{
$u,v\leftarrow{\bf RationalEndpoint}(w_1,w_{\ell},sr)$\label{algline:SubIBR-RationalEP}\;
add $(c,u,\ell)$ and $(v,d,\ell)$ to $stack$\;
}
}
\Return{$ans$}
\end{algorithm}

\begin{algorithm}[htbp]
\scriptsize
\DontPrintSemicolon
\LinesNumbered
\SetKwInOut{Input}{Input}
\SetKwInOut{Output}{Output}
\Input{
    $w_1(y)=F_1(r\pi+\arctan{y},y)$, where $F_1(x,y)\in\Q[x,y]$, $r\in\Q$\\
    $w_{\ell}(y)=F_{\ell}(r\pi+\arctan{y},y)$, where $F_{\ell}(x,y)\in\Q[x,y]$\\
    $sr=\Root(w_{\ell},c,d)$, $w_{\ell}$ has exactly one simple root $sr$ in $(c,d)$ and $w_{1}(sr)\ne0$
}
\Output{two rational numbers $u$ and $v$ $(c\leq u<sr<v\leq d)$ such that $w_{1}$ has no roots on $[u,v]$
}
\caption{{\bf RationalEndpoint}}\label{alg:RationalEndpoint}
\BlankLine
$n\leftarrow1$\;
\While{$n\ge1$}{
Compute an upper limit polynomial $W_{\max,n}\in\Q[y]$ and a lower limit polynomial $W_{\min,n}\in\Q[y]$ for $w_1(y)$\;
$q\leftarrow\frac{c+d}{2}$\;
\eIf{$(W_{\max,n}$ has no roots on $[c,d]$ and $W_{\max,n}(q)<0)$
{\bf or} \\$\quad\,(W_{\min,n}$ has no roots on $[c,d]$ and $W_{\min,n}(q)>0)$
}
{\Return{$c,d$}}
{
 $n\leftarrow n+1$\;
\uIf{$w_{\ell}(q)=0$}
{$c\leftarrow \frac{63q+c}{64},\;d\leftarrow\frac{63q+d}{64}$}
\uElseIf{$\sign(w_{\ell}(c)w_{\ell}(q))=-1$}
{$d\leftarrow q$}
\Else{$c\leftarrow q$}
}
}
\end{algorithm}

\section{The Whole Isolation Algorithm: Taking Root Multiplicity into Account}\label{sec:wholealgorithm}

In this section, we shall combine the results in previous sections to give an algorithm to isolate all real roots of an MTP. Furthermore, the multiplicity of each root is also computed by the algorithm. To do this, we shall first introduce a proposition on root multiplicity.
\begin{proposition}\label{prop:multiplicity-one}
    Let $g\in \mathbb{Q}[x,y]\backslash \mathbb{Q}$ be irreducible. Then except $x=0$, all real roots of $g(x,\tan x)$ in its domain are simple roots.
\end{proposition}
\begin{proof}
If $g\in\Q[x]\cup\Q[y]$, the conclusion is trivially true. When $g$ is bivariate the conclusion follows from \cite[Theorem 2.1]{chen2020automated}. Although the theorem only claims that any root of $g(x,\tan x)$ in $(0,\frac{\pi}{4}]$ is simple, its proof actually shows that any root in $(\mathbb{R}\backslash\{0\})\backslash(\mathbb{Z}+\frac{1}{2})\pi$ of $g(x,\tan x)$ is simple, which is just what we want.
\end{proof}
The following observation is similar to Lemma \ref{lemma:NoCommonRoot}. The proof given below is according to the proof of \cite[Theorem 7]{mccallum2012deciding}.

\begin{proposition}\label{proposition:nocommonroot}
Suppose that $p,q\in\mathbb{Q}[x,y]$ are co-prime, then there is no common root of $p(x,\tan x)$ and $q(x,\tan x)$ except $x=0$.
\end{proposition}
\begin{proof}
Replacing those ``$p$, $q$" in the proof of \cite[Theorem 7]{mccallum2012deciding} by $p$ and $q$ in this proposition, those ``$f^*(x)$, $g^*(x)$" therein by $p(x,\tan x)$ and $q(x,\tan x)$, and the function ``\text{trans}" therein by the tangent function, we see that if $\alpha\in(\mathbb{R}\backslash\{0\})\backslash(\mathbb{Z}+\frac{1}{2})\pi$ is a common root of $p(x,\tan x)$ and $q(x,\tan x)$, then both $\alpha$ and $\tan \alpha$ would be algebraic, which contradicts Theorem \ref{thm:BothAlgeNum}.
\end{proof}
Then the following theorem is straightforward.

\begin{theorem}\label{thm:counting-multiplicity}
    Let $g\in\mathbb{Q}[x,y]\backslash\mathbb{Q}$, and the factorization of $g$ in $\mathbb{Q}[x,y]$ is 
    $g=c_0\prod_{i=1}^{n} g_i^{p_i}$,
    where $g_i$ are distinct irreducible factors and $c_0\in\mathbb{Q}$. If $x_0\in\mathbb{R}\backslash\{0\}$ is a root of $g(x,\tan x)$, then there is a unique $g_i$ such that $g_i(x_0,\tan x_0)=0$ and the multiplicity of $x_0$ is $p_i$.\hfill\qed
\end{theorem}


    By the discussions above, Algorithms \ref{alg:AlgDerBound} and \ref{alg:BoundRSI} can be extended to compute the multiplicity of each root: Factorize $g$ into $c_0\prod_{i=1}^{n} g_i^{p_i}$ and call Algorithms \ref{alg:AlgDerBound} and \ref{alg:BoundRSI} for each irreducible factor $g_i$, then the multiplicity is just the exponent of $g_i$ in $g$ except at $x=0$. 
    
    \begin{remark}\label{remark:thesameforcot}
    Replacing the tangent function in Propositions \ref{prop:multiplicity-one}--\ref{proposition:nocommonroot} and Theorem \ref{thm:counting-multiplicity} by the cotangent function, the corresponding conclusions also hold. The proofs are all similar.
    \end{remark}
    
    The next lemma detects the possible roots of $f(x,\sin x,\cos x)$ at $k\pi~(k\in \mathbb{Z})$, of which the function $g(\frac{x}{2},\tan \frac{x}{2})$ does not take care.
\begin{lemma}\label{lemma:multiplicityofkpi}
    It can be effectively decided whether $f(x,\sin x,\cos x)$ is identically zero for $f\in\mathbb{Q}[x,y,z]$. Moreover, the root multiplicity of $f(x,\sin x,\cos x)$ at $0$, $2k\pi$ $(k\in \mathbb{Z}\backslash\{0\})$, $(2k+1)\pi$ $(k\in \mathbb{Z})$ can be effectively computed.
\end{lemma}
\begin{proof}
     Rewrite $f(x,\sin x,\cos x)$ as $g(\frac{x}{2},\tan\frac{x}{2})$ with $g\in\Q[x,t]$ in Eq. (\ref{equation:gTurnedToTan}), then $f(x,\sin x,\cos x)$ is identically zero iff $g(\frac{x}{2},\tan \frac{x}{2})$ is. Lemma 2.3 in \cite{chen2020automated} then shows $f(x,\sin x,\cos x)$ is identically zero iff $g$ is a zero polynomial.
     
     If $f(x,\sin x,\cos x)$ is not identically zero, then the root multiplicity of $x=0$ is finite. Otherwise the analytic function $f(z,\sin z,\cos z)$ with $z$ varying in $\C$ would be identically zero, and so would be $f(x,\sin x,\cos x)$. By computing $\frac{\mathrm{d}^{n} }{\mathrm{d} x^{n}}f(x,\sin x,\cos x)|_{x=0}$ for $n=0,1,2,...,$ until $\frac{\mathrm{d}^{n_0} }{\mathrm{d} x^{n_0}}f(x,\sin x,\cos x)|_{x=0}\neq0$ for some $n_0$. Then $n_0$ is the multiplicity of $x=0$.
     
    For $x=2k\pi$ with $k\in\Z\backslash\{0\}$, $f(x,\sin x,\cos x)|_{x=2k\pi}=0$ iff $g(x,\tan x)|_{x=k\pi}$\\$=0$ and the root multiplicity of $2k\pi$ \emph{w.r.t.} $f(x,\sin x,\cos x)$ equals the one of $k\pi$  \emph{w.r.t.} $g(x,\tan x)$. Moreover, we claim $g(k\pi,\tan k\pi)=0$ iff $t$ is a factor of $g$. The ``if" part is clear. Suppose $g(k\pi,\tan k\pi)=0$, then $0=h(k\pi,\tan k\pi)=h(k\pi,0)$ for some irreducible factor $h=\sum_{j=0}^{\deg(h,x)}a_j(t)x^j$ of $g$. Since $k\pi$ is transcendental, $a_j(0)=0$ for any $j$. Hence $t$ divides $h$ and $h=rt$ (with an $r\in\Q$), which proves the claim. Then, by Theorem \ref{thm:counting-multiplicity}, the multiplicity of $2k\pi$ is the exponent of $y$ in the factorization of $g$.
    
    Similarly, we can rewrite $f(x,\sin x,\cos x)$ as $\hat{g}(\frac{x}{2},\cot \frac{x}{2})$, with 
    \[\hat{g}(x,t)=f(2x,\frac{2t}{1+t^2},\frac{t^2-1}{1+t^2})(1+t^2)^{\text{deg}(f,[y,z])}\in\Q[x,t].\]
Then, using Remark \ref{remark:thesameforcot} and the discussion similar to the last paragraph, one observes that the root multiplicity of $(2k+1)\pi$ equals the exponent of $t$ in the factorization of $\hat{g}$.\end{proof}
    

It is time to give the main algorithm (Algorithm \ref{alg:CompleteAlgorithm}) for isolating all real roots of $f(x,\sin x,\cos x)$:
\begin{enumerate}
    \item Decide whether the MTP $f(x, \sin x,\cos x)$ is identically zero (Lines \ref{algline:whetherg0}--\ref{algline:error}).
    \item Compute the root multiplicity of the points $x=k\pi$ with $k\in\Z$ using the method in the proof of Lemma \ref{lemma:multiplicityofkpi} and compute a polynomial $g\in\Q[x,y]$ according to the input $f$ such that $y\nmid$ $g$ (Lines \ref{algline:m1}--\ref{algline:m3}).
    \item Compute the isolating intervals of $\lc(g,x)$ together with $(a_0,b_0)$, $(a_{s+1},b_{s+1})$, forming a potential periodic interval set $I$ and factorize $g$ (Lines \ref{algline:isolate}--\ref{algline:factorize}).
    \item Call Algorithm \ref{alg:AlgDerBound} with each irreducible factor $g_i$ with $\deg(g_i,y)\geq1$ together with $I$ to compute some $k^*_1$, $k^*_2$ and isolate the real roots of $g_i(x,\tan x)$ outside $(-k^*_2\pi-\pi/2,k^*_1\pi+\pi/2)$ (Lines \ref{algline:deggy>0part}--\ref{algline:renewk-}). 
    \item Isolate nonzero roots of the factors $g_i$ in $\Q[x]$ (Lines \ref{algline:isolateAlge}--\ref{algline:remove00}). Note that these correspond to the only nonzero algebraic roots of $f(x,\sin x,\cos x)$. We call $[0,0]$ and the intervals obtained in the above steps the ``algebraic" ones and others the ``transcendental" ones: the former represent 0 or algebraic roots and the latter represent transcendental roots. An algebraic interval may intersect a transcendental one, but it won't cause any trouble since we know they represent different roots.
    \item Unify those $k^*_{1\cdot 2}$ of each factor $g_i$ (with $\deg(g_i,y)\geq1$) by some $k_+$ and $k_-$, which characterize the bound of the ``bounded" roots (Lines \ref{algline:renewk+}, \ref{algline:renewk-} and \ref{algline:unifyk+}--\ref{algline:unifyk-}).
    \item For each factor $g_i$ such that $\deg(g_i,y)\geq1$, isolate the real roots of $g_i(x,\tan x)$ in $(k_-\pi-\pi/2,k_+\pi+\pi/2)$ by Algorithm \ref{alg:BoundRSI} (Lines \ref{algline:k-tok+}--\ref{algline:refineboundedroots}).
    \item Return all the results.
\end{enumerate}
In addition, for any real numbers $\alpha,\beta$, we define the interval $[[\alpha,\beta]]$ as
\begin{equation*}
[[\alpha,\beta]]:=\begin{cases}
(\alpha,\beta), &\text{if}~\alpha<\beta,\\
[\alpha,\beta], &\text{if}~\alpha=\beta.
\end{cases}  
\end{equation*}
The notation is used in the pseudocode of Algorithm \ref{alg:CompleteAlgorithm}.

\begin{algorithm}[!t]\label{alg:CompleteAlgorithm}
\tiny
\DontPrintSemicolon
\LinesNumbered
\SetKwInOut{Input}{Input}
\SetKwInOut{Output}{Output}
\Input{ 
  $f(x,\tau,\upsilon)\in \Q[x,\tau,\upsilon]$;\\
  a positive rational number $\varepsilon$ that specifies the maximal length of isolating intervals.
}
\Output{\\
a finite set of some $2$-tuples $\{([[\alpha_i,\beta_i]],m_i)\}_i$, where $\alpha_i,\beta_i\in (2\arctan(\Q)+\frac{\Z}{2}\pi)\cup\Q$, such that for every $i$, there is $1$ real root  of $f(x,\sin x,\cos x)$ in $[[\alpha_i,\beta_i]]$ and the root is of multiplicity $m_i$;\\
a finite set of some $3$-tuples $\{([[\alpha_i,\beta_i]],(c_i,m_i),k_i)\}_i $, where $\alpha_i,\beta_i\in 2\arctan(\Q)+\frac{\Z}{2}\pi$, such that for every $i$: if $k_i\ge 0$,  there are exactly $c_i$ roots of $f(x,\sin x,\cos x)$ in $[[\alpha_i,\beta_i]]+2k\pi\ (k\ge k_i)$ and each root is of multiplicity $m_i$;
if $k_i< 0$, there are exactly $c_i$ roots of $f(x,\sin x,\cos x)$ in $[[\alpha_i,\beta_i]]+2k\pi\ (k\le k_i)$ and each root is of multiplicity $m_i$.
}
\caption{\bf Main Algorithm}\label{alg:complete-algorithm}
\BlankLine
 $ans_1\leftarrow \emptyset$; $ans_{2}\leftarrow \emptyset$; $g\leftarrow f(2 x,\frac{2 y}{1+y^2},\frac{1-y^2}{1+y^2})\cdot (y^2+1)^{\tdeg(f,[\tau,\upsilon])}$\;
\If{\label{algline:whetherg0}$g$ is the zero polynomial}{
    \textbf{ERROR} ``$f(x,\sin x,\cos x)$ is identically zero" \;\label{algline:error}
}
 $gcot\leftarrow f(2 x,\frac{2 y}{1+y^2};  \frac{y^2 - 1}{1+y^2})\cdot (1 + y^2)^{\tdeg(f, [\tau,\upsilon])}$\;
Let $m_1$ be the maximal positive integer such that $y^{m_1}|gcot$\;\label{algline:m1}
\If{$m_1\neq 0$}{
$ans_2$.add($([\pi,\pi],(1,m_1),0)$); $ans_2$.add($([\pi,\pi],(1,m_1),-1)$)\;
}
\If{$f(0,\sin0,\cos0)=0$}{
    $m_2\leftarrow$ the multiplicity of $f(x,\sin{x},\cos x)$ at $0$; $ans_1$.add($([0,0],m_2)$)\;
}
Let $m_3$ be the maximal positive integer such that $y^{m_3}|g$\;
\If{$m_3\neq 0$}{
    $ans_2$.add($([0,0],(1,m_3),1)$); $ans_2$.add($([0,0],(1,m_3),-1)$); $g\leftarrow g/y^{m_3}$\;
}\label{algline:m3}
$(a_1,b_1),\ldots, (a_s,b_s)\leftarrow$ the isolating  intervals of $\lc(g,x)$ such that $\arctan b_i -\arctan a_i <\varepsilon/2~(1\le i\le s)$\;\label{algline:isolate}
$a_0\leftarrow-\infty$; Let $b_0\in\Q$ such that $\arctan b_0- (-\frac{\pi}{2})< \frac{\varepsilon}{2}$ and $b_0< a_1$.\; 
$b_{s+1}\leftarrow+\infty$; Let $a_{s+1}\in\Q$ such that $ \frac{\pi}{2}-\arctan a_{s+1}< \frac{\varepsilon}{2}$ and $a_{s+1}> b_s$.\; 
$I\leftarrow [(a_0,b_0),\ldots,(a_{s+1},b_{s+1})]$\;\label{algline:PPIS}
Factorize $g$ as ${c_0}\prod_{i=1}^{n} g_i^{p_i}$\;\label{algline:factorize}
$ans_{21}\leftarrow \emptyset;~ans_{22}\leftarrow \emptyset$; $k_-\leftarrow 0$; $k_+\leftarrow 0$\;
\For{$i$ from 1 to $n$}{
    \eIf{$\deg(g_i,y)\geq 1$}{\label{algline:deggy>0part}
        $k^*_1, \{((\alpha_{j},\beta_{j}),c_{j}) \}_j$ $\leftarrow$ Call Algorithm \ref{alg:AlgDerBound} with $g_i(x,y)$ and $I$\;
        $k_+\leftarrow\max\{k^*_1, k_+\}$; $ans_{21}$.add($\{((2\alpha_{j},2\beta_{j}),(c_{j},p_i)) \}_j$)\;\label{algline:renewk+}
        $k^*_2, \{((\alpha_{j},\beta_{j}),c_{j}) \}_j$ $\leftarrow$ Call Algorithm \ref{alg:AlgDerBound} with $g_i(-x,-y)$ and $I$\;
        $k_-\leftarrow\min\{-k^*_2,k_-\}$; $ans_{22}$.add($\{((-2\beta_{j},-2\alpha_{j}),(c_{j},p_i)) \}_j$)\;\label{algline:renewk-}
    }
    {
        $\{[[\alpha_{j},\beta_{j}]]\}_j\leftarrow$   the isolating  intervals of nonzero roots of ${g_i}$ such that they do not intersect $[0,0]$ and $\beta_{j} - \alpha_{j}<\varepsilon/2$.\;\label{algline:isolateAlge}
        $ans_1$.add($\{([[2\alpha_{j},2\beta_{j}]],p_i)\}_j$)\;\label{algline:remove00}
        
    }
}
Refine the intervals obtained in the ``\textbf{else}" part of the above loop so that they do not intersect.\;
\For{each $((\alpha^*,\beta^*),(c^*,p^*))$ in $ans_{21}$\label{algline:unifyk+}}{
    $ans_2$.add($((\alpha^*,\beta^*),(c^*,p^*),k_++1)$)\;
}
\For{each $((\alpha^*,\beta^*),(c^*,p^*))$ in $ans_{22}$}{
    $ans_2$.add($((\alpha^*,\beta^*),(c^*,p^*),k_--1)$)\;
}\label{algline:unifyk-}

\For{$k$ from $k_-$ to $k_+$}{\label{algline:k-tok+}
    \For{$i$ from 1 to $n$}{
        \If{$\deg(g_i,y)\geq 1$}{
            $\{[[\alpha_{j},\beta_{j}]]\}_j\leftarrow$ Call Algorithm \ref{alg:BoundRSI} with $g_i(x,y)$ and $k$\;
             Remove $[0, 0]$ in  $\{[[\alpha_{j},\beta_{j}]]\}_j$; 
          $ans_1$.add($\{([[2\alpha_{j},2\beta_{j}]],p_i)\}_j$)\;\label{algline:Remove0}
        }
    }
    Refine the intervals  $[[\alpha_{j},\beta_{j}]]$  obtained in the above loop so that any two of them do not intersect and $\beta_{j} - \alpha_{j}<\varepsilon/2$.\;\label{algline:refineboundedroots}
}

\Return  $ans_1$, $ans_2$\;
\end{algorithm}

\section{Experiments}\label{sec:experiments}

The main algorithm (Algorithm \ref{alg:complete-algorithm}) is implemented with {\tt Maple2021} as a tool named {\tt RootOfMTP}. 
In order to indicate the effectiveness and efficiency of the tool, we prepare two types of testing examples, where one type is collected from the literature
and another is generated randomly.
On the one hand, we test the correctness of the tool on these examples through some numerical methods.
On the other hand, the experimental results of these examples show that the tool is powerful and is able to deal with some complicated examples.

\subsection{Commands}
The tool {\tt RootOfMTP} can be downloaded at
\begin{center}
    {\bf \url{https://github.com/lihaokun/RootOfMTP}}.
\end{center}
All testing examples can also be downloaded there.
All tests were conducted on 16-Core Intel Core i7-12900KF@3.20GHz
with 128GB of memory and Windows 11.

We illustrate how to use {\tt RootOfMTP} by a simple example.
Suppose we want to isolate all the real roots of the following MTP by {\tt RootOfMTP}:
\begin{align*}
x\sin x + \cos x - 1.
\end{align*}
We only need to run the following commands in {\tt Maple2021}:

{\scriptsize
\begin{verbatim}
read ".../RootOfMTP.mpl";
RootOfMTP(x*sin(x)+cos(x)-1,x,1);
\end{verbatim}}

Herein, the first command is to read the file, and the inputs of the second command are an MTP, the variable of the MTP and a rational number $\epsilon>0$ which specifies the maximal length of isolating intervals.
The output is

{\scriptsize
\begin{verbatim}
For every k <= -1 (k in Z), 2kPi+(0) is a real root with multiplicity 1.
For every k >= 1 (k in Z), 2kPi+(0) is a real root with multiplicity 1.
There is 1 real root with multiplicity 2 at 0.
For every k >= 2 (k in Z), there is 1 real root with multiplicity 1 
in (2kPi+(2*arctan(63/16)), 2kPi+(Pi)).
For every k <= -2 (k in Z), there is 1 real root with multiplicity 1 
in (2kPi+(-Pi), 2kPi+(-2*arctan(63/16))).
There is 1 real root with multiplicity 1 in each open interval of the list
[[-5/2*Pi-2*arctan(29666650363354128505/36893488147419103232),
-5/2*Pi-2*arctan(7242537696610193/9007199254740992)],
[-1/2*Pi-2*arctan(7030038563941/17592186044416),
-1/2*Pi-2*arctan(14741934773129570377/36893488147419103232)],
[1/2*Pi+2*arctan(14741934773129570377/36893488147419103232),
1/2*Pi+2*arctan(7030038563941/17592186044416)],
[5/2*Pi+2*arctan(7242537696610193/9007199254740992),
5/2*Pi+2*arctan(29666650363354128505/36893488147419103232)]].
\end{verbatim}}

\subsection{Examples from the literature}
We collect $21$ testing examples from the literature \cite{chenliu2016,chen2022automated}.
Since some examples are rational-function inequalities, we need to transform them to MTPs.
Take \cite[Ex.4]{chenliu2016} for example:
\begin{align*}
    \frac{\sin^2 x}{x^2}+\frac{\tan x}{x}>\frac{x^2}{\sin^2 x}+\frac{x}{\tan x}.
\end{align*}
First, we move the terms in the right hand side to the left hand side. 
Second, we replace $\tan x$ with $\frac{\sin x}{\cos x}$.
Third, we cancel the denominators and obtain an MTP: $-x^4\cos x-x^3\sin x\cos^2 x+x\sin^3 x+\sin^4 x\cos x$.

The experimental results are presented in Table \ref{tab:chen-result}.
Note that these examples are kind of simple.
The bounds $k_{-}$ and $k_{+}$ are close to $0$ and {\tt RootOfMTP} can deal with each of them within $2$ seconds.

\begin{table}[ht]
\scriptsize
    \centering
    \scalebox{0.75}{
    \begin{tabular}{|c|c|}
\hline
$k_{-}$ & an integer such that $2k_{-}\pi-\pi$ is a lower bound for all bounded real roots\\
\hline
$k_{+}$ & an integer such that $2k_{+}\pi+\pi$ is an upper bound for all bounded real roots
\\
\hline
PRoots & periodic real roots\\
\hline
BRoots & bounded real roots\\
\hline
wSeq & the lengths of four computed weak Fourier sequences \\
\hline
wTime & time for computing four weak Fourier sequences\\
\hline
Time & total computing time \\
\hline
$n(m)$ & the number of bunches of periodic real roots / 
the number of bounded real roots (the multiplicity of each root)\\
\hline
-&no periodic or bounded real roots / no need to compute weak Fourier sequences\\
\hline
    \end{tabular}}
    \caption{Notation.} 
    \label{tab:notation}
\end{table}

\begin{table}[ht]
\scriptsize
    \centering
    \scalebox{0.55}{
    \begin{tabular}{|c|c|c|c|c|c|c|c||c|c|c|c|c|c|c|c|}
\hline
Example&$k_{-}$&$k_{+}$&PRoots&BRoots&wSeq&wTime&Time&Example&$k_{-}$&$k_{+}$&PRoots&BRoots&wSeq&wTime&Time\\
\hline
\cite[Ex.1]{chen2022automated}&-1&1&4(1)&7(1)&6;5;5;6&0.005s&0.251s&
\cite[Ex.2]{chen2022automated}&0&0&-&1(5)&6;5;5;6&0.006s&0.054s\\
\hline
\cite[Ex.4]{chen2022automated}&0&0&1(2)&-&-&0s&0.003s&
\cite[Ex.1, $k=1$]{chenliu2016} &0&0&4(1)&1(3)&8;7;7;8&0.009s&0.075s\\
\hline
\cite[Ex.1, $k=2$]{chenliu2016}&0&0&4(1)&1(5)&8;7;7;8&0.009s&0.058s&
\cite[Ex.1, $k=3$]{chenliu2016}&0&0&4(1)&1(3);2(1)&8;7;7;8&0.01s&0.085s\\
\hline
\cite[Ex.2, $c=9/45$]{chenliu2016}&-1&1&4(1)&1(6);6(1)&21;43;43;21&0.124s&1.089s&
\cite[Ex.2, $c=8/45$]{chenliu2016}&-1&1&4(1)&1(8);6(1)&21;43;43;21&0.084s&1.065s\\
\hline
\cite[Ex.2, $c=7/45$]{chenliu2016}&-1&1&4(1)&1(6);4(1)&21;43;43;21&0.123s&1.013s&
\cite[Ex.2, $c=6/45$]{chenliu2016}&-1&1&4(1)&1(6);4(1)&21;43;43;21&0.074s&0.964s\\
\hline
\cite[Ex.3]{chenliu2016}&-3&3&4(1)&1(5);12(1)&28;51;51;28&0.238s&1.719s&
\cite[Ex.4]{chenliu2016}&0&0&4(1)&1(8)&6;6;6;6&0.005s&0.246s\\
\hline
\cite[Ex.5]{chenliu2016}&0&0&-&2(1)&7;7;7;7&0.007s&0.098s&
\cite[Ex.6]{chenliu2016}&0&0&4(1)&1(7)&17;15;15;17&0.053s&0.18s\\
\hline
\cite[Ex.7]{chenliu2016}&0&0&-&1(5)&6;5;5;6&0.005s&0.048s&
\cite[Ex.8]{chenliu2016}&-1&1&4(1)&1(2);4(1)&4;3;3;4&0.001s&0.129s\\
\hline
\cite[Ex.9]{chenliu2016}&0&1&-&1(2)&5;5;5;5&0.005s&0.138s&
\cite[Ex.10]{chenliu2016}&0&0&4(1)&1(3)&8;7;7;8&0.056s&0.105s\\
\hline
\cite[Ex.11]{chenliu2016}&0&0&-&1(1)&4;4;4;4&0.003s&0.042s&
\cite[Ex.12]{chenliu2016}&-1&1&4(1)&8(1)&7;8;8;7&0.006s&0.187s\\
\hline
\cite[Ex.13]{chenliu2016}&0&0&4(1)&1(3)&8;7;7;8&0.007s&0.053s&&&&&&&&\\
\hline

    \end{tabular}}
    \caption{Results for Examples from the Literature (s: seconds).}
    \label{tab:chen-result}
\end{table}

\subsection{Random examples}

Now, we present the experimental results for randomly generated
polynomials. A randomly generated polynomial 
\begin{align*}
    {\bf rp}(d,p,c)\in\Q[x,y,z]
\end{align*}
is constructed as follows:
First, pick some monomials independently from the set 
$\{x^{a_1}y^{a_2}z^{a_3}\mid a_i\in\N,~a_1+a_2+a_3\le d\}$ in a manner such that each monomial is picked with probability $p$ while discarded with probability $1-p$. 
Then, randomly assign coefficients between $-c$ and $c$ to each picked monomial and sum them to obtain a polynomial.
We generate $160$ random polynomials $F_1,\ldots,F_{160}$ in this way from $16$ different classes, where
\begin{align*}
    F_{1},\ldots,F_{10}&\in{\bf rp}(5,0.1,10), & F_{11},\ldots,F_{20}&\in{\bf rp}(10,0.015,10),\\
    F_{21},\ldots,F_{30}&\in{\bf rp}(10,0.05,10), & F_{31},\ldots,F_{40}&\in{\bf rp}(10,0.1,10),\\
    F_{41},\ldots,F_{50}&\in{\bf rp}(15,0.01,10), & F_{51},\ldots,F_{60}&\in{\bf rp}(15,0.05,10),\\
    F_{61},\ldots,F_{70}&\in{\bf rp}(15,0.1,10), & F_{71},\ldots,F_{80}&\in{\bf rp}(20,0.0025,10),\\
    F_{81},\ldots,F_{90}&\in{\bf rp}(5,0.1,1000), & F_{91},\ldots,F_{100}&\in{\bf rp}(10,0.015,1000),\\
    F_{101},\ldots,F_{110}&\in{\bf rp}(10,0.05,1000), & F_{111},\ldots,F_{120}&\in{\bf rp}(10,0.1,1000),\\
    F_{121},\ldots,F_{130}&\in{\bf rp}(15,0.01,1000), & F_{131},\ldots,F_{140}&\in{\bf rp}(15,0.05,1000),\\
    F_{141},\ldots,F_{150}&\in{\bf rp}(15,0.1,1000), & F_{151},\ldots,F_{160}&\in{\bf rp}(20,0.0025,1000).
\end{align*}
Then, we replace $y$ and $z$ with $\sin x$ and $\cos x$ respectively in $F_{1},\ldots,F_{160}$ to get MTPs $M_{1},\ldots,M_{160}$.

The experimental results are presented in Table \ref{tab:random-result}.
The tool {\tt RootOfMTP} performs powerfully on these random examples: The tool can solve every example within $3000$ seconds, and $93\%$ of the examples were solved within 1 minute.
We notice that the computation of weak Fourier sequences is not very efficient.
For $59\%$ of the examples, the time for computing the weak Fourier sequences is more than one third of the total computing time. 
Another surprising observation is that the bounds $k_{-}$ and $k_{+}$ are still small, even on these more difficult examples.

Finally, we present Figure \ref{fig:num-within-time} to show the performance of {\tt RootOfMTP} on all examples (either from the literature or randomly generated).
The horizontal axis represents time, while the vertical axis represents the number of solved examples within the corresponding time.

\begin{table}[ht]
\scriptsize
    \centering
    \scalebox{0.45}{
    \begin{tabular}{|c|c|c|c|c|c|c|c||c|c|c|c|c|c|c|c|}
\hline
Example&$k_{-}$&$k_{+}$&PRoots&BRoots&wSeq&wTime&Time&Example&$k_{-}$&$k_{+}$&PRoots&BRoots&wSeq&wTime&Time\\
\hline
$M_{1}$&0&0&-&2(1)&27;27;27;27&0.186s&0.532s&$M_{2}$&0&0&4(1)&4(1)&26;26;26;26&0.148s&0.634s\\
\hline
$M_{3}$&0&0&4(1)&2(1)&22;24;24;22&0.08s&0.286s&$M_{4}$&-3&3&4(1)&15(1)&18;14;14;18&0.032s&0.592s\\
\hline
$M_{5}$&0&2&8(1)&11(1)&16;20;20;16&0.054s&0.542s&$M_{6}$&0&0&4(1)&4(1)&29;53;53;29&0.212s&0.62s\\
\hline
$M_{7}$&0&0&7(1)&2(1)&20;20;20;20&0.078s&0.361s&$M_{8}$&-2&2&4(1)&11(1)&49;26;26;49&0.215s&1.402s\\
\hline
$M_{9}$&-1&1&6(1)&7(1)&39;21;21;39&0.146s&0.878s&$M_{10}$&0&0&-&-&13;12;12;13&0.013s&0.101s\\
\hline
$M_{11}$&0&0&10(1)&5(1)&99;88;88;99&1.325s&2.167s&$M_{12}$&0&0&5(1);2(2)&1(1)&-&0s&0.018s\\
\hline
$M_{13}$&0&0&3(1)&4(1)&53;54;54;53&0.421s&1.367s&$M_{14}$&0&0&2(1)&1(4)&-&0s&0.009s\\
\hline
$M_{15}$&-1&1&4(1)&9(1)&87;192;192;87&2.765s&8.305s&$M_{16}$&-1&0&12(1)&15(1)&175;174;174;175&4.416s&13.173s\\
\hline
$M_{17}$&0&0&2(2)&1(2)&-&0s&0.001s&$M_{18}$&0&0&3(2)&1(3);1(1)&20;23;23;20&0.06s&0.238s\\
\hline
$M_{19}$&-5&1&10(1)&29(1)&70;212;212;70&3.87s&15.224s&$M_{20}$&0&0&3(8)&1(10)&-&0s&0.001s\\
\hline
$M_{21}$&0&1&-&2(1)&107;107;107;107&1.566s&7.036s&$M_{22}$&0&0&-&1(2);2(1)&79;79;79;79&1.001s&3.233s\\
\hline
$M_{23}$&-1&1&12(1)&16(1)&216;195;195;216&5.432s&23.089s&$M_{24}$&-4&4&4(1)&20(1)&104;213;213;104&4.035s&22.34s\\
\hline
$M_{25}$&0&0&8(1)&4(1)&107;243;243;107&5.278s&9.108s&$M_{26}$&-1&2&8(1)&16(1)&193;154;154;193&5.656s&19.36s\\
\hline
$M_{27}$&-2&2&8(1)&26(1)&143;161;161;143&4.508s&19.111s&$M_{28}$&0&0&-&1(2);1(1)&92;92;92;92&1.49s&2.893s\\
\hline
$M_{29}$&-1&1&8(1)&12(1)&108;177;177;108&3.516s&10.417s&$M_{30}$&0&1&4(1)&5(1)&99;233;233;99&4.659s&10.586s\\
\hline
$M_{31}$&0&0&4(1)&2(1)&103;247;247;103&5.227s&7.974s&$M_{32}$&0&1&4(1)&6(1)&105;249;249;105&5.022s&14.39s\\
\hline
$M_{33}$&-1&1&8(1)&12(1)&250;97;97;250&7.234s&19.997s&$M_{34}$&0&0&4(1)&4(1)&270;111;111;270&6.47s&11.022s\\
\hline
$M_{35}$&-1&1&8(1)&12(1)&535;111;111;535&55.14s&83.45s&$M_{36}$&0&0&4(1)&1(1)&259;106;106;259&8.499s&10.169s\\
\hline
$M_{37}$&0&1&4(1)&3(1)&103;260;260;103&7.776s&11.775s&$M_{38}$&0&0&4(1)&5(1)&100;100;100;100&1.575s&4.608s\\
\hline
$M_{39}$&0&1&4(1)&4(1)&103;103;103;103&2.115s&6.234s&$M_{40}$&-1&1&4(1)&7(1)&232;104;104;232&4.875s&12.368s\\
\hline
$M_{41}$&0&0&-&1(1)&74;90;90;74&1.325s&1.669s&$M_{42}$&-1&1&4(1)&1(5);10(1)&521;204;204;521&43.424s&84.569s\\
\hline
$M_{43}$&0&0&4(1)&1(6);3(1)&106;51;51;106&0.787s&1.46s&$M_{44}$&-6&6&8(1)&1(6);53(1)&463;382;382;463&44.619s&374.288s\\
\hline
$M_{45}$&0&0&4(1)&1(3);2(1)&183;611;611;183&97.672s&116.082s&$M_{46}$&-1&1&8(1)&1(2);13(1)&283;95;95;283&10.331s&29.126s\\
\hline
$M_{47}$&0&0&8(1)&2(1)&167;218;218;167&13.509s&17.624s&$M_{48}$&-1&1&10(1)&13(1)&336;538;538;336&70.131s&190.27s\\
\hline
$M_{49}$&0&0&2(1)&1(3)&84;91;91;84&1.13s&2.598s&$M_{50}$&-1&2&11(1)&17(1)&141;151;151;141&5.713s&19.468s\\
\hline
$M_{51}$&-2&1&8(1)&14(1)&200;905;905;200&230.777s&610.059s&$M_{52}$&0&0&-&1(2);2(1)&223;223;223;223&7.84s&19.67s\\
\hline
$M_{53}$&0&0&4(1)&1(2);5(1)&709;222;222;709&140.087s&173.988s&$M_{54}$&-1&1&8(1)&10(1)&605;787;787;605&233.809s&518.236s\\
\hline
$M_{55}$&0&1&-&2(1)&197;197;197;197&7.109s&25.48s&$M_{56}$&0&0&4(1)&1(4)&219;219;219;219&6.292s&16.779s\\
\hline
$M_{57}$&0&0&-&4(1)&223;223;223;223&7.016s&25.085s&$M_{58}$&-1&1&8(1)&1(2);14(1)&701;593;593;701&126.757s&368.925s\\
\hline
$M_{59}$&0&0&-&1(2);3(1)&230;230;230;230&10.187s&24.844s&$M_{60}$&-1&1&8(1)&12(1)&1133;219;219;1133&382.221s&806.921s\\
\hline
$M_{61}$&-1&0&4(1)&7(1)&245;650;650;245&60.368s&153.185s&$M_{62}$&0&0&-&4(1)&233;233;233;233&11.5s&32.585s\\
\hline
$M_{63}$&-1&1&4(1)&8(1)&243;617;617;243&50.81s&146.751s&$M_{64}$&-1&1&4(1)&7(1)&565;236;236;565&44.672s&97.59s\\
\hline
$M_{65}$&-7&7&4(1)&31(1)&234;537;537;234&43.514s&284.004s&$M_{66}$&0&1&4(1)&1(2);8(1)&237;597;597;237&52.063s&121.855s\\
\hline
$M_{67}$&-1&1&8(1)&13(1)&236;616;616;236&53.379s&171.749s&$M_{68}$&-1&1&4(1)&10(1)&242;631;631;242&50.009s&162.565s\\
\hline
$M_{69}$&-1&0&4(1)&5(1)&612;234;234;612&58.078s&91.915s&$M_{70}$&-1&0&4(1)&5(1)&609;238;238;609&52.149s&102.304s\\
\hline
$M_{71}$&0&0&3(1)&1(8);2(1)&164;164;164;164&3.571s&6.428s&$M_{72}$&-1&0&3(6);2(2);4(1)&1(13);6(1)&30;32;32;30&0.13s&5.488s\\
\hline
$M_{73}$&0&0&3(4);2(1)&1(6);3(1)&122;122;122;122&2.128s&4.072s&$M_{74}$&0&0&3(5);4(1)&1(7);5(1)&120;309;309;120&11.067s&23.141s\\
\hline
$M_{75}$&0&0&3(4)&1(4)&86;96;96;86&1.869s&2.906s&$M_{76}$&-1&0&15(1)&9(1)&1005;385;385;1005&338.574s&940.537s\\
\hline
$M_{77}$&-1&0&3(2);2(3);4(1)&1(4);2(1)&34;42;42;34&0.274s&0.819s&$M_{78}$&0&0&3(3);4(1)&1(7);3(1)&168;437;437;168&36.463s&45.945s\\
\hline
$M_{79}$&0&0&7(1)&1(2);1(1)&278;803;803;278&161.677s&170.291s&$M_{80}$&0&0&13(1)&1(4)&767;1009;1009;767&515.406s&518.845s\\
\hline
$M_{81}$&-1&0&-&4(1)&23;20;20;23&0.07s&0.379s&$M_{82}$&0&0&4(1)&2(1)&31;22;22;31&0.095s&0.298s\\
\hline
$M_{83}$&0&0&8(1)&3(1)&50;53;53;50&0.35s&0.787s&$M_{84}$&0&0&-&1(2);3(1)&17;18;18;17&0.104s&0.532s\\
\hline
$M_{85}$&0&0&4(1)&1(2);2(1)&16;14;14;16&0.04s&0.15s&$M_{86}$&0&0&-&3(1)&28;28;28;28&0.185s&0.739s\\
\hline
$M_{87}$&-1&1&11(1)&14(1)&17;18;18;17&0.049s&0.61s&$M_{88}$&0&0&-&1(2);1(1)&30;30;30;30&0.188s&0.588s\\
\hline
$M_{89}$&0&0&-&1(1)&19;19;19;19&0.109s&0.355s&$M_{90}$&-1&1&4(1)&6(1)&28;51;51;28&0.241s&1.035s\\
\hline
$M_{91}$&-1&0&-&1(2);1(1)&58;49;49;58&0.69s&2.015s&$M_{92}$&0&0&3(1)&1(7)&55;55;55;55&0.533s&1.269s\\
\hline
$M_{93}$&0&0&5(2);4(1)&1(3);2(1)&20;17;17;20&0.062s&0.243s&$M_{94}$&0&0&5(2)&1(6)&-&0s&0.002s\\
\hline
$M_{95}$&0&0&3(5)&1(9)&-&0s&0.0s&$M_{96}$&0&0&6(1)&1(6);1(1)&155;67;67;155&1.855s&2.914s\\
\hline
$M_{97}$&0&0&3(2);2(3)&1(5)&-&0s&0.006s&$M_{98}$&0&0&13(1)&5(1)&46;38;38;46&0.24s&0.667s\\
\hline
$M_{99}$&-10&10&6(1)&1(3);24(1)&5;5;5;5&0.003s&0.818s&$M_{100}$&0&0&4(1)&1(4);4(1)&28;48;48;28&0.205s&0.627s\\
\hline
$M_{101}$&0&0&-&1(2)&88;88;88;88&1.345s&3.151s&$M_{102}$&0&13&4(1)&1(3);78(1)&320;88;88;320&11.792s&95.276s\\
\hline
$M_{103}$&-7&7&8(1)&1(2);58(1)&327;87;87;327&15.814s&98.576s&$M_{104}$&-3&2&12(1)&32(1)&202;307;307;202&11.731s&66.217s\\
\hline
$M_{105}$&-1&0&7(1)&1(3);3(1)&94;222;222;94&4.1s&8.278s&$M_{106}$&-1&1&8(1)&9(1)&223;98;98;223&8.054s&16.822s\\
\hline
$M_{107}$&-1&0&8(1)&6(1)&335;242;242;335&14.961s&29.263s&$M_{108}$&0&0&8(1)&5(1)&179;175;175;179&8.592s&11.608s\\
\hline
$M_{109}$&-2&2&8(1)&21(1)&237;100;100;237&4.992s&28.19s&$M_{110}$&-1&0&-&1(1)&68;68;68;68&0.871s&3.634s\\
\hline
$M_{111}$&-1&1&4(1)&7(1)&110;245;245;110&5.441s&20.069s&$M_{112}$&-1&0&-&3(1)&82;82;82;82&1.066s&2.765s\\
\hline
$M_{113}$&-1&1&8(1)&13(1)&266;286;286;266&12.064s&42.789s&$M_{114}$&-1&0&4(1)&5(1)&107;245;245;107&6.711s&12.651s\\
\hline
$M_{115}$&-1&1&4(1)&7(1)&104;253;253;104&5.52s&15.223s&$M_{116}$&-1&4&8(1)&24(1)&271;271;271;271&14.857s&66.69s\\
\hline
$M_{117}$&-1&0&-&1(1)&104;104;104;104&1.941s&5.852s&$M_{118}$&-1&1&8(1)&8(1)&431;111;111;431&32.513s&49.383s\\
\hline
$M_{119}$&-1&0&4(1)&6(1)&111;279;279;111&11.232s&19.781s&$M_{120}$&0&0&-&-&107;107;107;107&1.763s&3.846s\\
\hline
$M_{121}$&-1&1&4(1)&8(1)&397;151;151;397&13.506s&39.123s&$M_{122}$&0&0&11(1)&1(10);4(1)&1443;182;182;1443&1273.551s&1766.71s\\
\hline
$M_{123}$&-6&0&17(1)&1(3);30(1)&833;543;543;833&267.967s&1832.919s&$M_{124}$&-68&1&6(1)&1(3);141(1)&151;623;623;151&88.491s&1008.551s\\
\hline
$M_{125}$&0&0&10(1)&3(1)&429;389;389;429&43.514s&65.364s&$M_{126}$&0&1&8(1)&6(1)&499;981;981;499&292.23s&870.745s\\
\hline
$M_{127}$&-1&1&2(2);4(1)&1(2);8(1)&311;129;129;311&25.591s&35.585s&$M_{128}$&0&0&-&3(1)&32;32;32;32&0.396s&0.928s\\
\hline
$M_{129}$&-1&1&8(1)&11(1)&132;132;132;132&2.595s&12.191s&$M_{130}$&-1&0&10(1)&1(3);10(1)&329;448;448;329&48.12s&93.598s\\
\hline
$M_{131}$&0&0&8(1)&5(1)&687;661;661;687&153.159s&249.382s&$M_{132}$&-1&0&4(1)&5(1)&225;560;560;225&40.575s&79.235s\\
\hline
$M_{133}$&0&0&-&4(1)&217;217;217;217&12.072s&23.128s&$M_{134}$&-1&1&8(1)&10(1)&513;1395;1395;513&1017.194s&2496.384s\\
\hline
$M_{135}$&-1&1&4(1)&6(1)&547;225;225;547&40.101s&99.09s&$M_{136}$&-1&1&8(1)&12(1)&649;533;533;649&106.612s&259.645s\\
\hline
$M_{137}$&0&1&4(1)&6(1)&623;241;241;623&58.741s&98.687s&$M_{138}$&-3&3&8(1)&1(2);28(1)&503;459;459;503&72.135s&325.427s\\
\hline
$M_{139}$&0&0&4(1)&4(1)&595;607;607;595&105.192s&138.137s&$M_{140}$&0&0&-&2(1)&215;215;215;215&8.132s&26.755s\\
\hline
$M_{141}$&0&0&4(1)&4(1)&242;621;621;242&60.554s&102.052s&$M_{142}$&-23&1&8(1)&18(1)&240;1361;1361;240&574.013s&1821.513s\\
\hline
$M_{143}$&-1&0&4(1)&6(1)&245;245;245;245&10.063s&55.621s&$M_{144}$&-1&1&4(1)&9(1)&227;586;586;227&51.978s&140.508s\\
\hline
$M_{145}$&0&0&-&1(1)&246;246;246;246&10.292s&32.055s&$M_{146}$&0&1&8(1)&8(1)&226;607;607;226&67.49s&115.328s\\
\hline
$M_{147}$&-4&4&4(1)&18(1)&247;647;647;247&55.513s&278.322s&$M_{148}$&-1&0&4(1)&4(1)&245;645;645;245&69.837s&131.675s\\
\hline
$M_{149}$&-1&0&8(1)&7(1)&234;630;630;234&73.299s&138.112s&$M_{150}$&0&1&12(1)&9(1)&1117;610;610;1117&542.28s&1267.414s\\
\hline
$M_{151}$&0&0&3(6);6(1)&1(13);2(1)&46;46;46;46&0.604s&1.933s&$M_{152}$&0&0&11(1)&1(4);2(1)&154;856;856;154&455.72s&663.41s\\
\hline
$M_{153}$&0&0&4(1)&1(3);1(1)&170;430;430;170&27.564s&43.251s&$M_{154}$&0&0&3(10);2(5)&1(12)&-&0s&0.023s\\
\hline
$M_{155}$&0&0&3(5);4(1)&1(5);2(1)&25;35;35;25&0.163s&0.568s&$M_{156}$&0&0&2(2);8(1)&4(1)&200;635;635;200&235.405s&275.278s\\
\hline
$M_{157}$&-1&1&4(1)&6(1)&433;213;213;433&15.891s&26.104s&$M_{158}$&0&0&2(3);4(1)&5(1)&94;138;138;94&3.508s&9.013s\\
\hline
$M_{159}$&-1&15&6(1)&64(1)&345;145;145;345&12.958s&168.032s&$M_{160}$&-1&1&11(1);2(2)&1(3);10(1)&203;279;279;203&15.479s&30.66s\\
\hline

    \end{tabular}}
    \caption{Results for Random Examples (s: seconds).}
    \label{tab:random-result}
\end{table}

\begin{figure}[ht]
    \centering
    \includegraphics[width=0.6\textwidth]{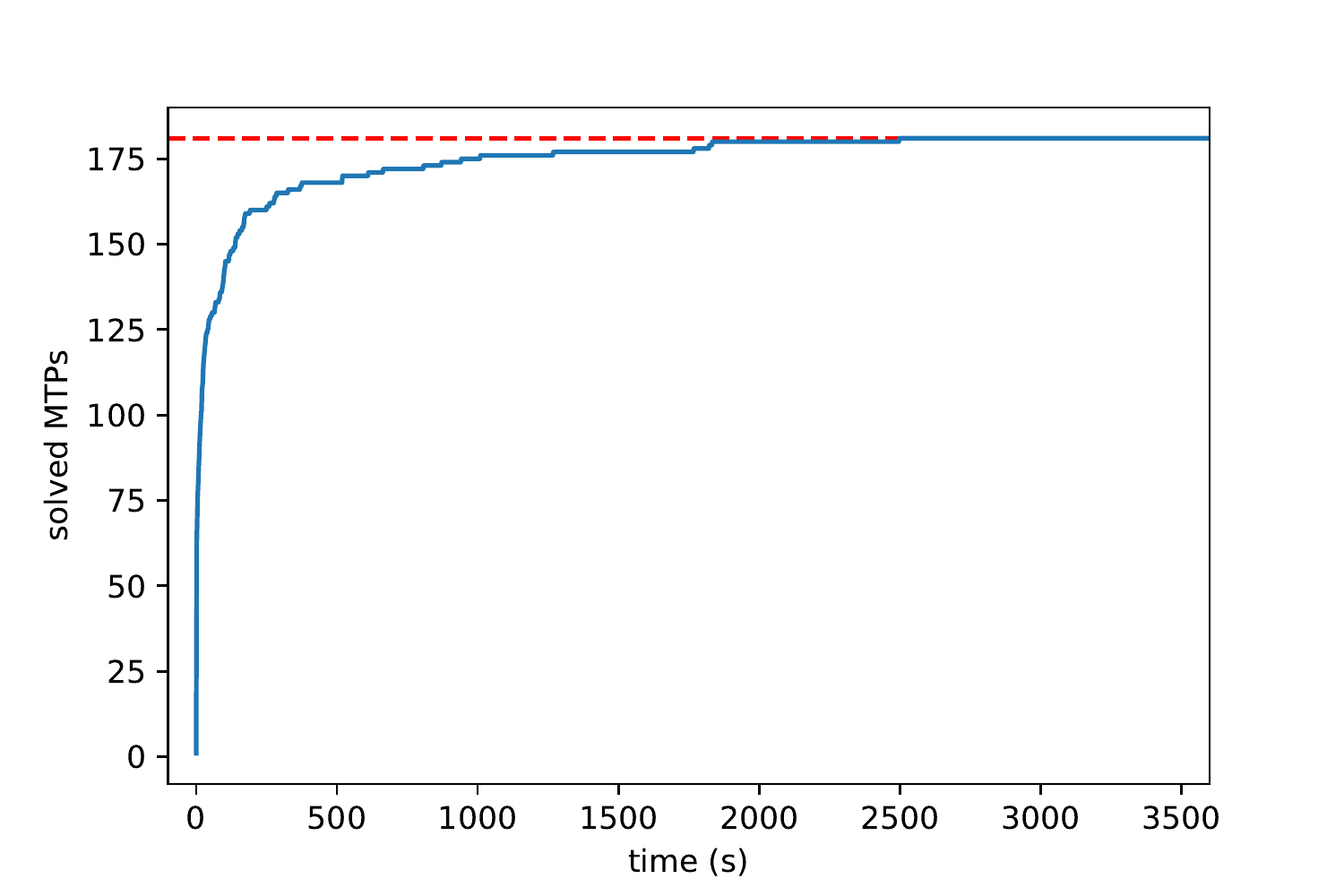}
    \caption{Number of Solved MTPs within Given Time (s: seconds). }
    \label{fig:num-within-time}
\end{figure}

\section{Conclusions}\label{sec:conclusion}
This paper proposes, for the first time, a complete algorithm ``isolating" all the real roots of a given MTP. More precisely, the algorithm computes
\begin{itemize}
    \item arbitrarily small intervals whose union covers all real roots of the given MTP (both the roots and the intervals can be countably many), and
    
    \item the exact number of distinct roots in each interval together with the multiplicity of each of them.
\end{itemize}
As a result, this leads to the first complete algorithm (as far as we know) deciding whether there is any, or whether there is infinitely many real roots for a given MTP. The correctness of the algorithm indicates a phenomenon of mathematical interest: the real roots in every period sufficiently far away from zero share a same quantity, a same multiset of root multiplicity and even similar relative positions in each period.

The algorithm has been implemented and a large number of examples have been tested to illustrate its efficiency.

Instead of focusing only on real roots in a certain bounded interval (as some other approaches in the literature doing), our (sub-)algorithms detect and ``isolate" the roots in any bounded intervals of the form $[2k_-\pi-\pi,2k_+\pi+\pi]$, and unbounded intervals of the forms $(2k_+\pi+\pi,\infty)$ and $(-\infty,-2k_-\pi-\pi)$ with sufficiently large $k_\pm$. Based on this, our methods can easily be modified to ``isolate" and compute the multiplicity of all real roots of an MTP in any interval $I$ of the form $(a,b)$, $[a,b)$, $(a,b]$ or $[a,b]$, with $a,b\in\mathbb{Q}\cup\{\pm\infty\}$: 
\begin{itemize}
    \item Choose an interval $I^\text{e}\supset I$ so that the sub-algorithms ``isolate" the real roots in $I^\text{e}$ naturally.
    \item If some of the ``transcendental" intervals cover some endpoints of $I$, replace them by some refined sub-intervals, each of which does not cover any endpoints but contains exactly one transcendental root.
    \item  After that, every ``isolating" interval covering some endpoints of $I$ is algebraic and contains exactly one root (which is $0$ or algebraic). If there are such intervals, then decide whether those covered endpoints are roots.
    \item If the covered endpoints are not roots, refine these endpoint-covering intervals via bisection so that they do not cover the endpoints after refinement. Otherwise, refine them to $[a,a]$ or $[b,b]$ if they cover $a$ or $b$, respectively.
\end{itemize}
In that way we can compute the exact number of real roots in each ``isolating" interval, together with multiplicity of each of them. Based on that, one can further prove MTP inequalities over the interval $I$ if one uses the following simple claim:
\begin{claim}
For any MTP $f$ not identically zero and any interval $I$ $($bounded or not$)$ of positive length,
\begin{itemize}
    \item $f>0$ on $I$ iff $f$ has no root in $I$ and $f(x_0)>0$ for some $x_0\in I$;
    \item $f\geq0$ on $I$ iff every root of $f$ in the interior of $I$ $($if there are any$)$ has even multiplicity and $f(x_0)>0$ for some $x_0\in I$.
\end{itemize}
\end{claim}
\noindent{}In this way an algorithm proving MTP inequalities over any interval (bounded or not) with endpoints in $\mathbb{Q}\cup\{\pm\infty\}$ can then be derived. To our knowledge, this is the first algorithm proving MTP inequalities over unbounded intervals.

\section*{Acknowledgement}
This work was supported by the NSFC under grants No. 61732001 and No. 12071467, and by the National Key Research Project of China under grant No. 2018YFA0306702.
\bibliographystyle{elsarticle-num}
\bibliography{ros}
\end{document}